\documentclass[12pt,onecolumn]{IEEEtran}
\usepackage{graphicx}
\usepackage{epsfig,amsfonts,amsmath,amssymb,amstext,amsthm}
\usepackage{threeparttable}
\usepackage{epstopdf}
\usepackage{footnote}
\usepackage{multirow}
\usepackage{setspace}

\newcommand{\bm}[1]{\mbox{\pmb{$#1$}}}

\newcommand{\dref}[1]{(\ref{#1})}
\newcommand{\Tr}{\mbox{Tr}}

\newcommand{\Perm}{\mbox{Perm}}

\newtheorem {Lemma}{Lemma}
\newtheorem {Corollary}{Corollary}

\begin{document}

\title{Ergodic Sum Capacity of Macrodiversity MIMO Systems in Flat Rayleigh Fading}

\author{Dushyantha A. Basnayaka,~\IEEEmembership{Student Member,~IEEE,}
        Peter J. Smith,~\IEEEmembership{Senior Member,~IEEE}
        and~Phillipa A. Martin,~\IEEEmembership{Senior Member,~IEEE}
\thanks{D. A. Basnayaka, P. J. Smith and P. A. Martin are with the Department of Electrical and Computer Engineering, University of Canterbury, Christchurch, New
Zealand. E-mail:\{dush, p.smith, p.martin\}@elec.canterbury.ac.nz.}
\thanks{D. A. Basnayaka is supported by a University of Canterbury International Doctoral Scholarship.}}

\maketitle
\begin{abstract}
The prospect of base station (BS) cooperation leading to joint
combining at widely separated antennas has led to increased interest
in macrodiversity systems, where both sources and receive antennas
are geographically distributed. In this scenario, little is known
analytically about channel capacity since the channel matrices have
a very general form where each path may have a different power.
Hence, in this paper we consider the ergodic sum capacity of a
macrodiversity MIMO system with arbitrary numbers of sources and
receive antennas operating over Rayleigh fading channels. For this
system, we compute the exact ergodic capacity for a two-source
system and a compact approximation for the general system, which is
shown to be very accurate over a wide range of cases. Finally, we
develop a highly simplified upper-bound which leads to insights into
the relationship between capacity and the channel powers. Results
are verified by Monte Carlo simulations and the impact on capacity
of various channel power profiles is investigated.
\end{abstract}


\section{Introduction}\label{sec:introduction}
With the advent of network multiple input multiple output (MIMO)
\cite{Siva07}, base station (BS) collaboration \cite{Fosch06} and
cooperative MIMO \cite{Big07}, it is becoming more common to
consider MIMO links where the receive array, transmit array or both
are widely separated. In these scenarios, individual antennas from a
single effective array may be separated by a considerable distance.
When both transmitter and receiver have distributed antennas, we
refer to the link as a macrodiversity MIMO link. Little is known
analytically about such links, despite their growing importance in
research \cite{Papa08}-\cite{Somekh07} and standards where
coordinated multipoint transmission (CoMP) is part of 3GPP LTE
Advanced.\\
Some analytical progress in this area has been made recently in the
performance analysis of linear combining for macrodiversity systems
in Rayleigh fading \cite{Dushp11, Dushdualuser11}. However, there
appears to be no work currently available on the capacity of general
systems of this type. Similar work includes the capacity analysis of
Rayleigh channels with a two-sided Kronecker correlation structure
\cite{Kiessling04}. However, the Kronecker structure is much too
restrictive for a macrodiversity layout and such results cannot be
leveraged here. Also, there is interesting work on system capacity
for particular cellular structures, including Wyner's circular
cellular array model \cite{Bacha06} and the infinite linear
cell-array model \cite{Somekh07}. Despite these contributions, the
general macrodiversity model appears difficult to handle. The
analytical difficulties are caused by the geographical separation of
the antennas which results in different entries of the channel
matrix having different powers with an arbitrary pattern. Also,
these powers can vary enormously when shadowing and path loss are
considered. Note that this type of channel model also occurs in the work of \cite{Werner06}.\\
In this paper, we consider a macrodiversity MIMO multiple access
channel (MIMO-MAC) where all sources and receive antennas are widely
separated and all links experience independent Rayleigh fading. For
this system, we consider the ergodic sum capacity, under the
assumption of no channel state information (CSI) at the
transmitters. For two sources, we derive the exact ergodic sum
capacity. The result is given in closed form, but the details are
complicated and for more than two sources, it would appear that an
exact approach is too complex to be useful. Hence, we develop an
approximation and a bound for the general case. The first technique
is very accurate, but the functional form is awkward to interpret.
Hence, a second, less accurate but simple bound is developed which
has a familiar and appealing structure. This bound leads to insight
into capacity behavior and its relationship with the channel powers.
In \cite{Dushisit12}, we presented a preliminary study of this
problem, which focussed on the approximation for the general case.
In this paper, we have extended the conference version to include
the exact two source results, correlated channels, full mathematical
details (see Sec. \ref{sec:general:preliminaries}), a motivation for
the approximate analysis (see Appendix
\ref{app:macro_capacity:extended_laplace}) and a much wider range of
scenarios, power profiles and discussion in the results section.\\
Note that, the methodology developed is for the case of arbitrary
powers for the entries in the channel matrix. There is no
restriction due to particular cellular structures. Hence, the
results and techniques may also have applications in multivariate
statistics.\\
The rest of the paper is laid out as follows. Section
\ref{sec:system_model} describes the system model and Sec.
\ref{sec:general:preliminaries} gives some mathematical
preliminaries required in the analysis. Section
\ref{sec:system_analysis_dual_user} provides an exact analysis for
the case of two source antennas. Sections
\ref{sec:system_analysis_general_user} and
\ref{sec:macro_capacity:simple_approximation} consider the case of
arbitrary numbers of sources and develop accurate approximations and
bounds on capacity. Results and conclusions appear in Secs.
\ref{sec:macro_capacity:numerical_analysis} and
\ref{sec:macro_capacity:conclusion}.
\section{System Model}\label{sec:system_model}
Consider a MIMO-MAC link with $M$ base stations and $W$ users
operating over a Rayleigh channel where BS $i$ has $n_{R_i}$ receive
antennas and user $i$ has $n_i$ antennas. The total number of
receive antennas is denoted $n_R=\sum_{i=1}^M n_{R_i}$ and the total
number of transmit antennas is denoted $N=\sum_{i=1}^W n_i$. An
example of such a system is shown in Fig.
\ref{fig:macro_capacity:1}, where three BSs are linked by a backhaul
processing unit (BPU) and communicate with multiple, mobile users.
All channels are considered to be independent since the correlated
channel scenario can be transformed into the independent case as
shown in Sec. \ref{sec:macro_capacity:correlated}. The system
equation is given by
\begin{align}
\bm{r} = \bm{Hs} + \bm{n},
\end{align}
\noindent where $\bm{r}$ is the $\mathcal{C}^{n_R \times 1}$ receive
vector, $\bm{s}$ is the combined $\mathcal{C}^{N \times 1}$
transmitted vector from the $W$ users, $\bm{n}$ is an additive white
Gaussian noise vector, $\bm{n} \sim \mathcal{CN} \left(\bm{0}
,\sigma^2 \bm{I}\right)$, and $\bm{H} \in \mathcal{C}^{n_R \times
N}$ is the composite channel matrix containing the $W$ channel
matrices from the $W$ users. The ergodic sum capacity of the link
depends on the availability of channel state information (CSI) at
the transmitter side. In particular, if no CSI at the transmitter is
assumed, the corresponding ergodic sum capacity is
\cite[pp.~57]{Big07}
\begin{align}
E \left \{ C \right \} &= E \left \{ \log_2 \left| \bm{I} +
\frac{1}{\sigma^2} \bm{HH}^H \right| \right\},
\label{eq:macro_capacity:1}
\end{align}
\begin{figure}[h]
\centerline{\includegraphics*[scale=0.80]{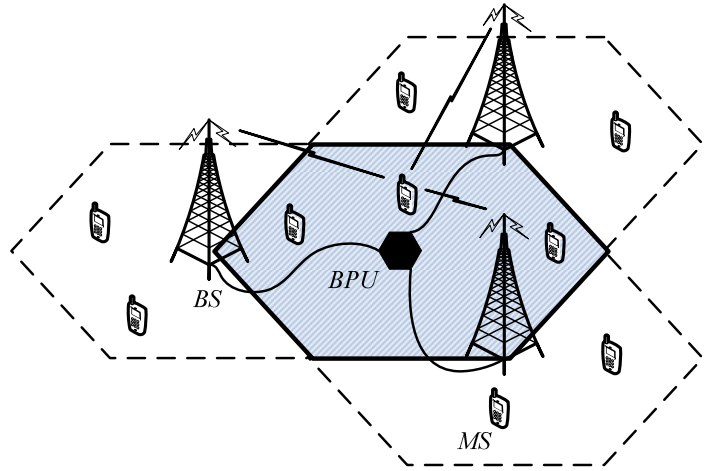}} \caption{A
network MIMO system with a 3 sector cluster. To reduce the clutter,
only paths from a single source are shown.}
\label{fig:macro_capacity:1}
\end{figure}
where $E \left \{ |s_i|^2 \right \}=1$, $i=1,2, \dots, N$, is the
power of each transmitted symbol. It is convenient to label each
column of $\bm{H}$ as $\bm{h}_i$, $i=1,2, \dots, N$, so that
$\bm{H}=\left( \bm{h}_1, \bm{h}_2, \dots , \bm{h}_N \right)$. The
covariance matrix of $\bm{h}_k$ is defined by $\bm{P}_k = E \left \{
\bm{h}_k \bm{h}_k^H \right \}$ and $\bm{P}_k=\mbox{diag}
\left(P_{1k}, P_{2k}, \dots, P_{n_Rk}\right)$. Hence, the $ik^{th}$
element of $\bm{H}$ is $\mathcal{CN} \left(0, P_{ik}\right)$. Using
this notation, we can also express $\bm{h}_k$ as
$\bm{h}_k=\bm{P}^{\frac{1}{2}}_k \bm{u}_k$, where $\bm{u}_k \sim
\mathcal{CN} \left(\bm{0}, \bm{I} \right)$. Note that, for
convenience, all the power information is contained in the
$\bm{P}_k$ matrices so that there is no normalization of the channel
and, in \dref{eq:macro_capacity:1}, the scaling factor in the
capacity equation is simply $1/\sigma^2$.
\subsection{Correlated Channels}\label{sec:macro_capacity:correlated}
Consider the general scenario where sources and/or BSs have multiple
co-located antennas for transmission and reception. Here, spatial
correlation may be present due to the co-located antennas
\cite{Shiu00, Tulino05}. If a Kronecker correlation model is
assumed, then the composite channel matrix is given by
\begin{align}\small
\!\!\!\!\bm{H} \!&=\! \left(\!
\begin{smallmatrix}
\bm{R}_{r1}^{\frac{1}{2}} &  & \bm{0} \\
\  & \ddots &  \\
\bm{0} & & \bm{R}_{rM}^{\frac{1}{2}}  \\
\end{smallmatrix}  \!\right)\!\!
\left(\!
\begin{smallmatrix}
\bm{H}_{w,11} & \ldots & \bm{H}_{w,1W}  \\
\vdots & \ddots & \vdots \\
\bm{H}_{w,M1} & \ldots & \bm{H}_{w,MW}  \\
\end{smallmatrix}  \!\right)\!\!
\left(\!
\begin{smallmatrix}
\bm{R}_{t1}^{\frac{1}{2}} &  & \bm{0}  \\
& \ddots &  \\
\bm{0} &  & \bm{R}_{tW}^{\frac{1}{2}}  \\
\end{smallmatrix}  \! \right),
\label{eq:macro_capacity:correlated_H1}
\end{align}
\noindent where the $\mathcal{C}^{n_{R_k} \times n_i}$ matrix,
$\bm{H}_{w,ik}$, has iid elements since all the channel powers from
user $k$ to BS $i$ are the same. The matrix $\bm{R}_{ri}$ is the
receive correlation matrix at BS $i$ and the matrix $\bm{R}_{tk}$ is
the transmit correlation matrix at source $k$ as defined in
\cite{Shiu00}. Using the spectral decompositions,
$\bm{R}_{ri}=\bm{\Phi}_{ri} \bm{\Lambda}_{ri} \bm{\Phi}_{ri}^H$ and
$\bm{R}_{tk}=\bm{\Phi}_{tk} \bm{\Lambda}_{tk} \bm{\Phi}_{tk}^H$ and
substituting \dref{eq:macro_capacity:correlated_H1} into
\dref{eq:macro_capacity:1} it is easily shown that the capacity with
the channel in \dref{eq:macro_capacity:correlated_H1} is
statistically identical to the capacity with channel
\begin{align}\small
\!\!\!\bm{H} \!&=\! \left(\!
\begin{smallmatrix}
\bm{\Lambda}_{r1}^{\frac{1}{2}} &  & \bm{0} \\
\  & \ddots &  \\
\bm{0} & & \bm{\Lambda}_{rM}^{\frac{1}{2}}  \\
\end{smallmatrix}  \!\right)\!\!
\left(\!
\begin{smallmatrix}
\bm{H}_{w,11} & \ldots & \bm{H}_{w,1W}  \\
\vdots & \ddots & \vdots \\
\bm{H}_{w,M1} & \ldots & \bm{H}_{w,MW}  \\
\end{smallmatrix}  \!\right)\!\!
\left(\!
\begin{smallmatrix}
\bm{\Lambda}_{t1}^{\frac{1}{2}} &  & \bm{0}  \\
& \ddots &  \\
\bm{0} &  & \bm{\Lambda}_{tW}^{\frac{1}{2}}  \\
\end{smallmatrix}  \! \right).
\label{eq:macro_capacity:correlated_H2}
\end{align}
Denoting \dref{eq:macro_capacity:correlated_H2} by
$\bm{H}=\bm{\Lambda}_{r}^{\frac{1}{2}} \bm{H}_{w}
\bm{\Lambda}_{t}^{\frac{1}{2}}$, we see that correlation is
equivalent to a scaling of the channel by the relevant eigenvalues
in $\bm{\Lambda}_{r}$ and $\bm{\Lambda}_{t}$. In particular, the
$\left(u,v \right)^{th}$ element of $\bm{H}$ has power
$\bm{\Lambda}_{r,uu} \bm{\Lambda}_{t,vv} P_{uv}$, where $ P_{uv}$ is
the single link power from transmit antenna $v$ to receive antenna
$u$. Hence, correlation can be handled by the same methodology
developed in Secs.
\ref{sec:system_analysis_dual_user}-\ref{sec:macro_capacity:simple_approximation},
with suitably scaled power values \footnote{Arbitrary fixed transmit
power control techniques can also be handled in the same way as for
the correlated scenario.}.
\section{Preliminaries}\label{sec:general:preliminaries}
In this section we derive some useful results which will be used
extensively throughout the paper.
\begin{Lemma}
Let $\bm{X}$ be an $n\times n$ complex random matrix with,
\begin{align}
\bm{A} \! &=\! E \left\{ \bm{X}\circ \bm{X}\right\} \triangleq  \! \left(
\!
\begin{array}{ccc} \vspace{1mm}
E \left\{ |X_{11}|^2 \right\} & E \left\{ |X_{12}|^2 \right\} & \ldots  \\
\vspace{1mm}
E \left\{ |X_{21}|^2 \right\} & E \left\{ |X_{22}|^2 \right\} & \ldots \\
\vspace{1mm}
\ldots & \ldots & \ldots  \\
E \left\{ |X_{n1}|^2 \right\} & E \left\{ |X_{n2}|^2 \right\} & \ldots  \\
\end{array} \! \right)\!,
\label{eq:general:expected_matrix_definition}
\end{align}
\noindent where $\circ$ represents the Hadamard product. With this
notation, the following identity holds.
\begin{align}
E \left\{ \left|\bm{X}^H\bm{X}\right|\right\} &= \mbox{perm} \left(
\bm{A} \right),
\end{align}
where $perm(.)$ is the permanent of a square matrix defined in
\cite{Minc78}.
\label{lemma:macro_capacity:expected_square_determinant}
\end{Lemma}

\begin{proof}
From the definition of the determinant of a generic matrix, $\bm{X}
= \left \{ X_{i,k}\right \}_{i,k=1\dots n}$, we have
\begin{align}
\begin{split}
E \left\{ \left|\bm{X}^H\bm{X}\right| \right\} = E& \left  \{
\left[\sum_\sigma
\mbox{sgn}(\sigma) \prod_{i=1}^n \bar{X}_{\sigma_i,i}\right] \right.
 \times  \left. \left[\sum_\mu \mbox{sgn}(\mu) \prod_{i=1}^n
X_{\mu_i,i}\right] \right\},
\end{split}
\label{lemma:macro_capacity:expected_square_determinant1}
\end{align}
where $\sigma = \left( \sigma_1, \sigma_2, \dots, \sigma_n \right)$
is a permutation of the integers $1, \dots, n$, the sum is over all
permutations, and $\mbox{sgn}(\sigma)$ denotes the sign of the
permutation. The permutation, $\mu$, in the second summation of
\dref{lemma:macro_capacity:expected_square_determinant1} is defined similarly.
Since all the elements of $\bm{X}$ are independent, the
only terms giving non-zero value expectations are $\prod_{i=1}^n
\bar{X}_{a_i,i}\prod_{i=1}^n X_{b_i,i}$, where permutation $a=b$.
Hence, using the permanent definition in \cite{Minc78} we have
\begin{align}
\begin{split}
E \left\{ \left|\bm{X}^H\bm{X}\right| \right\} &= \sum_\sigma
\prod_{i=1}^n A_{\sigma_i,i} = \mbox{perm}\left(\bm{A}\right).
\end{split}
\end{align}
\end{proof}
\begin{Corollary}
Let $\bm{X}$ be an $m\times n$ random matrix with, $E \left\{
\bm{X}\circ \bm{X}\right\} = \bm{A}$, where $\bm{A}$ is an $m\times n$
deterministic matrix and $m>n$. Then, the following identity holds.
\begin{align}
\begin{split}
E \left\{ \left|\bm{X}^H\bm{X}\right| \right\} = \mbox{Perm} \left(
\bm{A} \right),
\end{split}
\end{align}
where $\mbox{Perm} \left( \bm{A} \right)$ is the permanent of the
rectangular matrix $\bm{A}$ as defined in \cite{Minc78}.
\label{corollary:macro_capacity:expected_rectangular_determinant}
\end{Corollary}
\begin{proof}
Using the Cauchy-Binet formula for the determinant of the product of
two rectangular matrices, we can expand
$\left|\bm{X}^H\bm{X}\right|$ as a sum of products of two square
matrices. Each product of square matrices can be evaluated using
Lemma \ref{lemma:macro_capacity:expected_square_determinant}. The
resulting expression is seen to be the permanent of the rectangular
matrix, $\bm{A}$, which completes the proof.
\end{proof}
\begin{Corollary}
Let $\bm{X}$ be an $m\times n$ random matrix with, $E \left\{
\bm{X}\circ \bm{X}\right\} = \bm{A}$, where $\bm{A}$ is an $m\times n$
deterministic matrix and $m>n$. If the $m\times m$ deterministic
matrix $\bm{\Sigma}$ is diagonal, then the following identity holds.
\begin{align}
\begin{split}
E \left\{ \left|\bm{X}^H\bm{\Sigma}\bm{X}\right| \right\} =
\mbox{Perm} \left( \bm{\Sigma} \bm{A} \right).
\end{split}
\end{align}
\label{corollary:macro_capacity:expected_diagonal_rectangular_determinant}
\end{Corollary}
\begin{proof}
The result follows directly from Lemma
\ref{lemma:macro_capacity:expected_square_determinant}, and the fact
that $\bm{\Sigma}^{\frac{1}{2}} \circ \bm{\Sigma}^{\frac{1}{2}}=
\bm{\Sigma}$ for any diagonal matrix.
\end{proof}
Next, we give a definition for the elementary symmetric function
(esf) of degree $k$ in $n$ variables, $X_1, X_2, \dots, X_n$
\cite{Muir82}. Let $e_k \left(X_1, X_2, \dots, X_n \right)$ be the
$k^{th}$ degree esf, then
\begin{align}
e_k \left(X_1, X_2, \dots, X_n \right) &= \sum_{1 \leq l_1 < l_2 <
\dots < l_k \leq n} X_{l_1} \dots X_{l_k}.
\label{identity:macro_capacity:esf}
\end{align}
It is apparent from \dref{identity:macro_capacity:esf} that $e_0
\left(X_1, X_2, \dots, X_n \right) = 1$ and $e_n \left(X_1, X_2,
\dots, X_n \right) = X_1X_2\dots X_n$. In general, the esf of degree
$k$ in $n$ variables, for any $k\leq n$, is formed by adding
together all distinct products of $k$ distinct variables.
\begin{Lemma}
\cite{Muir82} let $\bm{X}$ be an $n \times n$ complex symmetric
positive definite matrix with  eigenvalues $\lambda_1, \dots,
\lambda_n$. Then, the following identity holds.
\begin{align}
e_k \left(\lambda_1, \lambda_2, \dots, \lambda_n \right) &=
\textup{\Tr}_k \left( \bm{X} \right),
\end{align}
\noindent where
\begin{align}
\textup{\Tr}_k \left( \bm{X} \right) = \begin {cases} \sum_{\sigma}
\left|
\bm{X}_{\sigma_{k,n}} \right| & \quad 1\leq k \leq n \\
1 & \quad k=0 \\
0 & \quad k > n,
\end {cases}
\end{align}
\label{lemma:macro_capacity:esf_trace_identity}
\end{Lemma}
\noindent where $\sigma_{k,n}$ is an ordered subset of $\left \{
n\right\}=\left\{1, \dots, n \right\}$ of length $k$ and the
summation is over all such subsets.
$\left.\bm{X}\right._{\sigma_{k,n}}$ denotes the principal submatrix
of $\bm{X}$ formed by taking only the rows and
columns indexed by $\sigma_{k,n}$.\\
In general, $\left.\bm{X}\right._{\sigma_{\ell,n}}^{\mu_{\ell,n}}$
denotes the submatrix of $\bm{X}$ formed by taking only the rows and
columns indexed by $\sigma_{\ell,n}$ and $\mu_{\ell,n}$
respectively, where $\sigma_{\ell,n}$ and $\mu_{\ell,n}$ are length
$\ell$ subsets of $\{1, 2, \dots ,n\}$. If either $\sigma_{\ell,n}$
or $\mu_{\ell,n}$ contains the complete set (i.e.,
$\sigma_{\ell,n}=\{1, 2, \dots ,n\}$ or $\mu_{\ell,n}=\{1, 2, \dots
,n\}$), the corresponding subscript/superscript may be dropped. When
$\sigma_{\ell,n}=\mu_{\ell,n}$, only one
subscript/superscript may be shown for brevity.\\
Next, we present three axiomatic identities for permanents which are required
in Sec. \ref{sec:system_analysis_general_user}.
\begin{itemize}
\item \emph{Axiom 1}: Let $\bm{A}$ be an arbitrary $m\times n$ matrix, then
\begin{align}
\sum_{\sigma} \Perm \left( \left(\bm{A} \right)^{\mu_{0,n}} \right)
&= \sum_{\sigma} \Perm \left( \left(\bm{A} \right)_{\sigma_{0,m}}
\right) = 1. \label{eq:macro_capacity:axiom1}
\end{align}
\item \emph{Axiom 2}: Let $\bm{A}$ be an arbitrary $m\times n$ matrix, then
\begin{align}
\sum_{\sigma} \Perm \left( \left( \bm{A}\right)_{\sigma_{k,m}}
\right) &= \sum_\sigma \Perm \left( \left(
\bm{A}\right)^{\sigma_{k,n}} \right).
\label{eq:macro_capacity:axiom2}
\end{align}
\item \emph{Axiom 3}: For an empty matrix, $\bm{A}$,
\begin{align}
\Perm \left(\bm{A} \right)&=1. \label{eq:macro_capacity:axiom3}
\end{align}
\end{itemize}
\section{Exact Small System Analysis}\label{sec:system_analysis_dual_user}
In this section, we derive the exact ergodic sum capacity in
\dref{eq:macro_capacity:1} for the $N=2$ case. This corresponds to
two single antenna users or a single user with two distributed
antennas.
Here, the channel matrix becomes $\bm{H}=\left( \bm{h}_1, \bm{h}_2
\right)$ and it is straightforward to write
\dref{eq:macro_capacity:1} as
\begin{align}
E \left \{ C \right \} \ln 2 &= E \left \{ \ln \left| \bm{I} +
\frac{1}{\sigma^2} \bm{h}_1 \bm{h}_1^H \right|\right \} \nonumber \\
&+\! E \left \{ \ln \left| \bm{I} \!+\! \frac{1}{\sigma^2} \left(
\bm{I} \!+\! \frac{1}{\sigma^2}\bm{h}_1 \bm{h}_1^H
\right)^{-1}\!\!\!\bm{h}_2 \bm{h}_2^H \right|\right \} \nonumber \label{eq:macro_capacity:2}\\
&\triangleq C_1 + C_2.
\end{align}
Both $C_1$ and $C_2$ can be expressed as scalars \cite{Paulraj03},
\cite[pp.~48]{Helmut96}, so the capacity analysis simply requires
\begin{align}
C_1 &= E \left \{ \ln \left( 1  + \frac{1}{\sigma^2} \bm{h}_1^H
\bm{h}_1 \right) \right \} \label{eq:macro_capacity:C1}, \\
C_2 &= E \left \{ \ln \left( 1 \!+\! \frac{1}{\sigma^2} \bm{h}_2^H
\left( \bm{I} \!+\! \frac{1}{\sigma^2}\bm{h}_1 \bm{h}_1^H
\right)^{-1}\!\!\!\bm{h}_2  \right) \right \}.
\label{eq:macro_capacity:C2}
\end{align}
In order to facilitate our analysis, it is useful to avoid the
logarithm representations in \dref{eq:macro_capacity:C1} and
\dref{eq:macro_capacity:C2}. We exchange logarithms for exponentials
as follows. First, we note the identity,
\begin{align}
\frac{1}{a} &= \int_0^\infty e^{-at}dt, \quad \mbox{for} \quad a >
0. \label{eq:macro_capacity:gamma_identity1}
\end{align}
\noindent Now equation \dref{eq:macro_capacity:gamma_identity1} can be
used to find $\ln a$ as below:
\begin{align}
\frac{\partial \ln a }{\partial a } &= \int_0^\infty e^{-at}dt, \\
\int_0^{\ln a} d \ln a &= \int_1^a \int_0^\infty e^{-a t} dt da, \\
\ln a &= \int_0^\infty \frac{e^{-t}  - e^{-at}}{t} dt.
\label{eq:macro_capacity:exponential_identity}
\end{align}
This manipulation is useful because there are many results which can
be applied to exponentials of quadratic forms, whereas few results
exist for logarithms. As an example, using
\dref{eq:macro_capacity:exponential_identity} in
\dref{eq:macro_capacity:C1} gives
\begin{align}
C_1 &= E \left \{ \int_0^\infty  \frac{e^{-t} -
e^{-\left(1+\frac{1}{\sigma^2} \bm{h}_1^H \bm{h}_1 \right)t}}{t} dt
\right \}. \label{eq:macro_capacity:C1_20}
\end{align}
Note that $a=1+\frac{1}{\sigma^2} \bm{h}_1^H \bm{h}_1$ has been used
in \dref{eq:macro_capacity:exponential_identity}. Sice $a \geq 1$,
it follows that the integrand in \dref{eq:macro_capacity:C1_20} is
non-negative. Also, the expected value, $C_1$, is clearly finite and
so, by Fubini's theorem, the order of expectation and integration in
\dref{eq:macro_capacity:C1_20} may be interchanged.
 Using the Gaussian integral identity \cite{Dushdualuser11}, the
expectation in \dref{eq:macro_capacity:C1_20} can be computed to
give
\begin{align}
C_1 &= \int_0^\infty  \frac{e^{-t}}{t} -
\frac{e^{-t}}{t \left| \bm{\Sigma}_1 \right|} dt,
\label{eq:macro_capacity:C1_2}
\end{align}
\noindent where $\bm{\Sigma}_1 = \bm{I} + \frac{t}{\sigma^2}
\bm{P}_1$. Hence, the log-exponential conversion in \dref{eq:macro_capacity:exponential_identity}
leads to a manageable integral for $C_1$. Using the same approach and applying \dref{eq:macro_capacity:exponential_identity} in \dref{eq:macro_capacity:C2} gives
\begin{align}
C_2 &= E \left \{ \int_0^\infty \frac{e^{-t}}{t} - \frac{ e^{-t -
\frac{t}{\sigma^2}\bm{h}_2^H \left( \bm{I} +
\frac{1}{\sigma^2}\bm{h}_1 \bm{h}_1^H \right)^{-1}\bm{h}_2 } }{t} dt
\right \}. \label{eq:macro_capacity:C2_1}
\end{align}
The expectation in \dref{eq:macro_capacity:C2_1} has to be
calculated in two stages. First, the expectation over $\bm{h}_2$ can
be solved using the Gaussian integral identity \cite{Dushdualuser11} and,
with some simplifications, we arrive at
\begin{align}
C_2 \!&=\! \int_0^\infty \!\! \frac{e^{-t}}{t}  - E_{\bm{h}_1}
\!\left \{ \frac{e^{-t} \left( \sigma^2 + \bm{h}_1^H
\bm{h}_1\right)}{t \left| \bm{\Sigma}_2 \right| \left( \sigma^2 +
\bm{h}_1^H \bm{\Sigma}_2^{-1} \bm{h}_1 \right)} \right \}dt ,
\label{eq:macro_capacity:C2_2}
\end{align}
\noindent where $\bm{\Sigma}_2 = \bm{I} + \frac{t}{\sigma^2}
\bm{P}_2$. Interchange of the expectation and integral in
\dref{eq:macro_capacity:C2_2} follows from the same arguments used
for $C_1$. Equation \dref{eq:macro_capacity:C2_2} can be further
simplified to give
\begin{align}
C_2 = &\int_0^\infty  \frac{e^{-t}}{t}- \frac{e^{-t}}{t \left| \bm{\Sigma}_2 \right|} dt
-  E_{\bm{h}_1} \!\left \{\! \frac{1}{\sigma^2} \int_0^\infty
\frac{e^{-t} \bm{h}_1^H \bm{P}_2 \bm{\Sigma}_2^{-1} \bm{h}_1 }{t
\left| \bm{\Sigma}_2 \right| \left( \sigma^2 + \bm{h}_1^H
\bm{\Sigma}_2^{-1} \bm{h}_1 \right)} dt \right \}.
\label{eq:macro_capacity:C2_3}
\end{align}
Defining the third term in \dref{eq:macro_capacity:C2_3} as $I_b$,
the ergodic sum capacity, $E \left( C \right)= C_1 + C_2$, becomes
\begin{align}
E \left \{ C \right \}  &= \frac{1}{\ln 2} \left \{ \sum_{k=1}^{2}
I_{a_k} - I_b \right \},
\end{align}
where
\begin{align}
I_{a_k} &=  \int_0^\infty \left( \frac{e^{-t}}{t} -
\frac{e^{-t}}{t\left| \bm{\Sigma}_k \right|} \right) dt.
\label{eq:macro_capacity:I_aj1}
\end{align}
Substituting for $\bm{\Sigma}_k$ in \dref{eq:macro_capacity:I_aj1}
and expanding $\left(t \left|\bm{\Sigma}_k\right|\right)^{-1}$ gives
\begin{align}
I_{a_k} &= \sum_{i=1}^{n_R} \eta_{ik} \int_0^\infty \frac{e^{-t}}{t
+ \frac{\sigma^2}{P_{ik}}} dt, \label{eq:macro_capacity:I_aj2}
\end{align}
\noindent where
\begin{align}
\eta_{ik} &= \frac{P_{ik}^{n_R-1}}{\prod_{l \neq i}^{n_R} \left(
P_{ik} - P_{lk}\right)}.\label{eq:macro_capacity:varphi1}
\end{align}
Note that the first $e^{-t}/t$ term in \dref{eq:macro_capacity:I_aj1} cancels out
with one of the terms in the partial fraction expansion leaving only the
linear terms shown in the denominator of \dref{eq:macro_capacity:I_aj2}.
The integrals in \dref{eq:macro_capacity:I_aj2} can be solved in
closed form \cite{GradRzy00} to give
\begin{align}
I_{a_k} &= \sum_{i=1}^{n_R} \eta_{ik} e^{\frac{\sigma^2}{P_{ik}}}
E_1 \left(\frac{\sigma^2}{P_{ik}}\right).
\label{eq:macro_capacity:I_aj3}
\end{align}
In order to compute $I_b$ we use \cite[Lemma~1]{Dushdualuser11} to give
\begin{equation}
I_{b} = -  \int_0^\infty \int_0^\infty
\frac{e^{-t}}{\left|\bm{\Sigma}_2 \right|} \left. \frac{\partial E
\left \{ e^{-\theta_1 z_1 - \theta_2 z_2} \right \} }{\partial
\theta_1}\right|_{\theta_1=0} d\theta_2 dt ,
\label{eq:macro_capacity:I_b1}
\end{equation}
\noindent where $z_1=\bm{h}_1^H \bm{P}_2 \bm{\Sigma}_2^{-1}
\bm{h}_1$ and $z_2=\sigma^2 + \bm{h}_1^H \bm{\Sigma}_2^{-1}
\bm{h}_1$. The expectation in \dref{eq:macro_capacity:I_b1} can be
solved as in \cite{Dushdualuser11}, and with some manipulations we arrive at
\begin{equation}
I_{b} = -\!\! \int_0^\infty \!\!\!\! \int_0^\infty \!\!
\frac{\partial}{\partial \theta_1 } \!\! \left[ \frac{e^{-\sigma^2t
- \sigma^2 \theta_2}}{\left| \bm{I} + t \bm{P}_2 + \theta_1
\bm{P}_1\bm{P}_2
 + \theta_2 \bm{P}_1 \right|} \right]_{\theta_1=0} \!\!\! d\theta_2 dt. \label{eq:macro_capacity:I_b2}
\end{equation}
In Appendix \ref{app:macro_capacity:A}, $I_b$ in
\dref{eq:macro_capacity:I_b2} is calculated in closed form and the
final result is given by
\begin{align}
I_b &=-  \frac{1}{\left| \bm{P}_1 \bm{P}_2 \right|} \left \{
\sum_{i=1}^{n_R} \sum_{k \neq i }^{n_R} \sum_{l \neq i,k }^{n_R}
\frac{\xi_{ikl}\left( \tilde{M}_{b_{ikl}} - \tilde{N}_{b_{ikl}}
\right)}{J_i} \right \}, \label{eq:macro_capacity:I_b4}
\end{align}
\noindent where $\tilde{M}_{b_{ikl}}$, $\tilde{N}_{b_{ikl}}$, $J_i$
and $\xi_{ikl}$ are given in \dref{eq:macro_capacity:I1_bikl3},
\dref{eq:macro_capacity:I2_bikl3}, \dref{eq:macro_capacity:jacobian}
and \dref{eq:macro_capacity:xi1} respectively. Then, the final result becomes
\begin{align}
\begin{split}
E \left \{ C \right \}  &= \frac{1}{\ln 2} \left \{ \sum_{k=1}^{2}
\sum_{i=1}^{n_R} \eta_{ik} e^{\frac{\sigma^2}{P_{ik}}} E_1
\left(\frac{\sigma^2}{P_{ik}}\right) \right. + \left.
\frac{1}{\left| \bm{P}_1 \bm{P}_2 \right|} \left [ \sum_{i=1}^{n_R}
\sum_{k \neq i }^{n_R} \sum_{l \neq i,k }^{n_R}
\frac{\xi_{ikl}\left( \tilde{M}_{b_{ikl}} - \tilde{N}_{b_{ikl}}
\right)}{J_i} \right ] \right \}.
\end{split}
\label{eq:macro_capacity:dual_user}
\end{align}
\section{Approximate General Analysis}\label{sec:system_analysis_general_user}
In this section, we present an approximate ergodic sum rate capacity
analysis for the case where $n_R \geq N > 2$. Extending this to $N
\geq n_R$ is a simple extension of the current analysis. We use the
following notation for the channel matrix,
\begin{subequations}
\begin{align}
\bm{H} &= \left( \tilde{\bm{H}}_N, \bm{h}_N \right) \\
&= \left( \tilde{\bm{H}}_{N-1}, \bm{h}_{N-1}, \bm{h}_N \right) \\
&=  \left( \tilde{\bm{H}}_k , \bm{h}_k \dots, \bm{h}_{N-1}, \bm{h}_N
\right)\\
&= \hspace{10mm} \vdots  \nonumber \\
&= \left( \bm{h}_1, \bm{h}_2, \dots , \bm{h}_N \right),
\label{eq:macro_capacity:general:notation}
\end{align}
\end{subequations}
\noindent where the $n_R \times \left(k-1\right)$ matrix,
$\tilde{\bm{H}}_k$, comprises the $k-1$ columns to the left of
$\bm{h}_k$ in $\bm{H}$. Using the same representation as in
\dref{eq:macro_capacity:2}, the ergodic sum capacity is defined by
\cite{Big07} as,
\begin{align}
E \left \{ C \right \} \ln 2 &\triangleq \sum_{k=1 }^N C_k,
\label{eq:macro_capacity:general:3}
\end{align}
\noindent where
\begin{align}
C_k \!&=\! E \left \{  \ln \left| \bm{I} \!+\! \frac{1}{\sigma^2}
\left( \bm{I} \!+\! \frac{1}{\sigma^2} \tilde{\bm{H}}_k
\tilde{\bm{H}}_k^H
\right)^{-1} \!\!\! \bm{h}_k \bm{h}_k^H \right|\right \} \label{eq:macro_capacity:general:1}\\
&=  E \left \{  \ln \left( 1 + \frac{1}{\sigma^2} \bm{h}_k^H \left(
\bm{I} + \frac{1}{\sigma^2} \tilde{\bm{H}}_k \tilde{\bm{H}}_k^H
\right)^{-1} \!\!\! \bm{h}_k \right) \right \}.
\label{eq:macro_capacity:general:2}
\end{align}
Applying \dref{eq:macro_capacity:exponential_identity} to
\dref{eq:macro_capacity:general:2} gives
\begin{align}
C_k \!&=\! \int_0^\infty \!\! \frac{e^{-t}}{t} - E \left \{
\frac{e^{-t} \left| \sigma^2 \bm{I} + \tilde{\bm{H}}_k^H
\tilde{\bm{H}}_k \right|}{t \left| \bm{\Sigma}_k \right| \left|
\sigma^2 \bm{I} + \tilde{\bm{H}}_k^H \bm{\Sigma}_k^{-1}
\tilde{\bm{H}}_k \right|} \right \}dt,
\label{eq:macro_capacity:general:Cj1}
\end{align}
\noindent where $\bm{\Sigma}_k = \bm{I} + \frac{t}{\sigma^2}
\bm{P}_k$. In order to calculate the second term in
\dref{eq:macro_capacity:general:Cj1}, the following expectation
needs to be calculated,
\begin{align}
\tilde{I}_{k} \left( t \right) &=  \frac{1}{\left| \bm{\Sigma}_k
\right|} E \left\{ \frac{\left| \sigma^2 \bm{I} + \tilde{\bm{H}}_k^H
\tilde{\bm{H}}_k \right|}{\left| \sigma^2 \bm{I} +
\tilde{\bm{H}}_k^H \bm{\Sigma}_k^{-1} \tilde{\bm{H}}_k \right|}
\right \}. \label{eq:macro_capacity:general:ICj1}
\end{align}
Exact analysis of $\tilde{I}_{k} \left( t \right)$ is cumbersome,
and even the $N=2$ case (see the $I_b$ calculation in
\dref{eq:macro_capacity:I_b2}) is complicated. Hence, we employ a
Laplace type approximation \cite{Lib94}, so that $\tilde{I}_{k}
\left( t \right)$ can be approximated by
\begin{align}
\tilde{I}_{k} \left( t \right) &\simeq  \frac{1}{\left|
\bm{\Sigma}_k \right|}  \frac{ E \left\{ \left| \sigma^2\bm{I} +
\tilde{\bm{H}}_k^H \tilde{\bm{H}}_k \right| \right \} }{ E \left\{
\left| \sigma^2 \bm{I} + \tilde{\bm{H}}_k^H \bm{\Sigma}_k^{-1}
\tilde{\bm{H}}_k \right| \right\} }.
\label{eq:macro_capacity:general:ICj2}
\end{align}
Note that the Laplace approximation is better known for ratios of
scalar quadratic forms \cite{Lib94}. However, the identity in both
the numerator and denominator of
\dref{eq:macro_capacity:general:ICj1} can be expressed as the limit
of a Wishart matrix as in \cite{Gao98}. This gives
\dref{eq:macro_capacity:general:ICj1} as ratio of determinants of
matrix quadratic forms which in turn can be decomposed to give a
product of scalar quadratic forms as in Appendix
\ref{app:macro_capacity:extended_laplace} and \cite{mythesis}.
Hence, the Laplace approximation for
\dref{eq:macro_capacity:general:ICj1} has some motivation in the
work of \cite{Lib94}. It can also be thought of as a first order
delta expansion \cite{StOr94}. From Appendix
\ref{app:macro_capacity:numerator_expectation}, the expectation in
the numerator of \dref{eq:macro_capacity:general:ICj2} is given by
\begin{equation}
E \left\{ \left| \sigma^2\bm{I} + \tilde{\bm{H}}_k^H
\tilde{\bm{H}}_k \right| \right \} \! =\! \sum_{i=0}^{k-1}
\sum_{\sigma} \Perm \left ( \left(\bm{Q}_k \right)^{\sigma_{i,k-1}}
\right ) \left( \sigma^2 \right)^{k-i-1},
\label{eq:macro_capacity:Cj_numerator:expectation}
\end{equation}
\noindent where $\bm{Q}_k$ is defined in
\dref{eq:macro_capacity:Qk1}. From Appendix
\ref{app:macro_capacity:denom_expectation}, the expectation in the
denominator of \dref{eq:macro_capacity:general:ICj2} is given by
\begin{align}
\left|\bm{\Sigma}_k\right| E \left\{ \left|\sigma^2 \bm{I} +
\tilde{\bm{H}}_k^H \bm{\Sigma}_k^{-1} \tilde{\bm{H}}_k\right| \right
\} &= \sum_{l=0}^{n_R} \left. t \right.^{l} \varphi_{kl},
\label{eq:macro_capacity:ICj2_denominator:expectation1}
\end{align}
\noindent where $\varphi_{kl}$ is given in
\dref{eq:app:macro_capacity:ICj2_denominator:expectation:varphi1}.
Therefore, $\tilde{I}_{k} \left( t \right)$ becomes
\begin{align}
\tilde{I}_{k} \left( t \right) &\simeq \frac{\Theta \left(\bm{Q}_k
\right)}{\sum_{l=0}^{n_R} \left. t\right.^{l}
\varphi_{kl}} \\
&= \frac{\Theta \left(\bm{Q}_k \right)}{\varphi_{kn_R}
\sum_{l=0}^{n_R} \left( \frac{\varphi_{kl}}{\varphi_{kn_R}}\right)
\left. t\right .^{l}} \label{eq:macro_capacity:general:ICj30} \\
&= \frac{\Theta \left(\bm{Q}_k \right)}{\varphi_{kn_R}
\prod_{l=1}^{n_R} \left(t + \omega_{kl} \right) },
\label{eq:macro_capacity:general:ICj3}
\end{align}
where
\begin{align}
\Theta \left(\bm{Q}_k \right) = \sum_{i=0}^{k-1} \sum_{\sigma} \Perm
\left( \left(\bm{Q}_k \right)^{\sigma_{i,k-1}} \right)
\left(\sigma^2 \right)^{k-i-1}.
\label{eq:macro_capacity:general:theta}
\end{align}
Note that $\omega_{kl} > 0$ for all $l,k$ from Descartes' rule of
signs as all the coefficients of the monomial in the denominator of
\dref{eq:macro_capacity:general:ICj30} are positive. Also note that,
from
\dref{eq:app:macro_capacity:ICj2_denominator:expectation:varphi1},
we have $\Theta \left(\bm{Q}_k \right)=\varphi_{k0}$.
Applying \dref{eq:macro_capacity:general:ICj3} in
\dref{eq:macro_capacity:general:Cj1} gives
\begin{align}
C_k &\simeq \int_0^\infty \!\! \frac{e^{-t}}{t}  -
\frac{\varphi_{k0}}{\varphi_{kn_R}} \frac{e^{-t}}{t
\prod_{l=1}^{n_R} \left(t + \omega_{kl} \right)} dt.
\label{eq:macro_capacity:general:Cj2}
\end{align}
Using a partial fraction expansion for the product in the denominator
of the second term of \dref{eq:macro_capacity:general:Cj2} gives
\begin{align}
\frac{1}{{t \prod_{l=1}^{n_R} \left(t + \omega_{kl} \right)}} &=
\frac{\zeta_{k0}}{t} - \sum_{l=1}^{n_R} \frac{\zeta_{kl}}{\left. t +
\omega_{kl} \right. },
\label{eq:macro_capacity:general:partial_fraction_Cj1}
\end{align}
\noindent where
\begin{align}
\zeta_{k0} &= \frac{1}{\prod_{u = 1 }^{n_R} \omega_{ku} } =
\frac{\varphi_{kn_R}}{\varphi_{k0}}
\label{eq:macro_capacity:general:partial_fraction_Cj_zeta0}
\end{align}
and
\begin{align}
\zeta_{kl} &= \frac{1}{\omega_{kl} \prod_{u \neq l }^{n_R} \left(
\omega_{ku} - \omega_{kl} \right) }.
\label{eq:macro_capacity:general:partial_fraction_Cj_zeta1}
\end{align}
Applying
\dref{eq:macro_capacity:general:partial_fraction_Cj1} in
\dref{eq:macro_capacity:general:Cj1} gives
\begin{align}
C_k &\simeq \frac{\varphi_{k0}}{\varphi_{kn_R}} \sum_{l=1}^{n_R}
\int_0^\infty \frac{\zeta_{kl}}{\left. t +
\omega_{kl} \right. } dt \label{eq:macro_capacity:general:Cj3} \\
&=  \frac{\varphi_{k0}}{\varphi_{kn_R}} \sum_{l=1}^{n_R} \zeta_{kl}
e^{\omega_{kl}} E_1 \left( \omega_{kl} \right).
\label{eq:macro_capacity:general:Cj4}
\end{align}
\noindent Then, applying \dref{eq:macro_capacity:general:Cj4} in
\dref{eq:macro_capacity:general:3} gives the final approximate
ergodic sum rate capacity as
\begin{align}
\!\!\! E \left \{ C \right \} &\triangleq  \frac{1}{\ln 2}\sum_{k=1
}^N \left(  \frac{\varphi_{k0}}{\varphi_{kn_R}} \sum_{l=1}^{n_R}
\zeta_{kl} e^{\omega_{kl}} E_1 \left(\omega_{kl} \right) \right ).
\label{eq:macro_capacity:general:4}
\end{align}
Note the simplicity of the general approximation in
\dref{eq:macro_capacity:general:4} in comparison to the two-user
exact results in \dref{eq:macro_capacity:dual_user}.
\section{A Simple Capacity Bound}\label{sec:macro_capacity:simple_approximation}
In this section, we derive an extremely simple upper-bound for the
ergodic capacity in \dref{eq:macro_capacity:1}. This provides a
simpler relationship between the average link powers and ergodic sum
capacity at the expense of a loss in accuracy. Using Jensen's
inequality \cite{Tom91} and $\bar{\gamma}=\frac{1}{\sigma^2}$,
equation \dref{eq:macro_capacity:1} leads to
\begin{align}
E \left \{ C  \right \}  &\leq  \log_2 \left( E \left \{ \left|
\bm{I} + \bar{\gamma} \bm{H}^H \bm{H} \right| \right \} \right).
\label{eq:macro_capacity:capacity_bound2}
\end{align}
From Appendix \ref{app:macro_capacity:numerator_expectation},
\dref{eq:macro_capacity:capacity_bound2} can be given as
\begin{align}
E \left \{ C  \right \}  &\leq  \log_2 \left( \sum_{i=0}^{N}
\sum_{\sigma} \Perm \left(\bm{P}^{\sigma_{i,N}}\right)
\bar{\gamma}^{i} \right), \label{eq:macro_capacity:capacity_bound3}
\\
&= \log_2 \left( \sum_{i=0}^{N} \vartheta_i \bar{\gamma}^{i}
\right).\label{eq:macro_capacity:capacity_bound31}
\end{align}
\noindent where $\bm{P}=\left(P_{ik}\right)$. The simplicity of
\dref{eq:macro_capacity:capacity_bound3} is hidden by the permanent
form. For small systems, expanding the permanent reveals the simple
relationship between the upper bound and the channel powers. For
$n_R=N=2$ and $n_R=N=3$, \dref{eq:macro_capacity:capacity_bound31}
gives the upper bounds in
\dref{eq:macro_capacity:capacity_bound:2by2} and
\dref{eq:macro_capacity:capacity_bound:3by3} respectively. These
bounds show the simple pattern where cross products of $L$ channel
powers scale the $\bar{\gamma}^L$ term. Hence, at low SNR where the
$\bar{\gamma}$ term is dominant, $\log_2\left( 1 + P_T \bar{\gamma}
\right)$, where $P_T=\sum_i \sum_k P_{ik}$, is an approximation to
\dref{eq:macro_capacity:capacity_bound31}. Similarly, at high SNR,
the $\bar{\gamma}^N$ term is dominant and $\log_2\left( 1 + \Perm
\left(\bm{P}\right)\bar{\gamma}^N \right)$ is an approximation.
These approximations show that capacity is affected by the sum of
the channel powers at low SNR, whereas at high SNR, the cross
products of $N$ powers becomes important.
\begin{table*}[!t]
\small
%
\begin{align}
E \left \{ C  \right \}  &\leq \log_2 \left( 1 + \bar{\gamma} \left(
P_{11} + P_{12} + P_{21} + P_{22}\right) + \bar{\gamma}^{2} \left(
P_{11}P_{22} + P_{12}P_{21} \right) \right).
\label{eq:macro_capacity:capacity_bound:2by2}
\end{align}
\begin{align}
\begin{split}
E \left \{ C  \right \}  \leq \log_2 &\left( 1 + \bar{\gamma} \left(
P_{11} + P_{12} + P_{13} + P_{21} + P_{22} + P_{23} + P_{31} +
P_{32} + P_{33} \right) \right. \\
&+ \left. \bar{\gamma}^{2} \left( P_{11}P_{22} + P_{11}P_{32} +
P_{21}P_{12} + P_{21}P_{32} + P_{31}P_{12} + P_{31}P_{22} +
P_{11}P_{23} + P_{11}P_{33} + P_{21}P_{13}  \right. \right. \\
&+ \left. \left.
P_{21}P_{33}+P_{31}P_{13}+P_{31}P_{23}+P_{12}P_{23}+P_{12}P_{33}+P_{22}P_{13}+P_{22}P_{33}+P_{32}P_{13}+P_{32}P_{23}
\right) \right. \\
&+ \left.  \bar{\gamma}^{3} \left( P_{11}P_{22}P_{33} +
P_{11}P_{23}P_{32} + P_{12}P_{21}P_{33} + P_{12}P_{31}P_{23} +
P_{13}P_{21}P_{32} + P_{13}P_{22}P_{31} \right) \right).
\end{split}
\label{eq:macro_capacity:capacity_bound:3by3}
\end{align}
\hrulefill
\end{table*}

\section{Numerical and Simulation Results}\label{sec:macro_capacity:numerical_analysis}
For the numerical results, we consider three distributed BSs with
either a single receive antenna or two antennas. For a two-source
system, we parameterize the system by three parameters
\cite{Dushp11, Dushdualuser11}. The average received signal to noise
ratio is defined by $\rho=P_T / \sigma^2$. In particular for a
two-source system, $\rho= \left( \Tr\left( \bm{P}_1\right) +
\Tr\left( \bm{P}_2 \right)\right) /\sigma^2$. The total signal to
interference ratio is defined by $\varsigma = \Tr\left(
\bm{P}_1\right) / \Tr\left( \bm{P}_2\right)$. The spread of the
signal power across the three BS locations is assumed to follow an
exponential profile, as in \cite{Gao98}, so that a range of
possibilities can be covered with only one parameter. The
exponential profile is defined by
\begin{align} P_{ik} &= K_k \left( \alpha \right) \alpha^{i-1},
\end{align}
for receive location $i \in \left \{1,2,3 \right \}$ and source $k$
where
\begin{align}
K_k \left( \alpha \right) &= \Tr\left( \bm{P}_k\right) /
\left(1+\alpha + \alpha^2 \right), \quad k=1,2,
\end{align}
and $\alpha > 0$ is the parameter controlling the uniformity of the
powers across the antennas. Note that as $\alpha \rightarrow 0$ the
received power is dominant at the first location, as $\alpha$
becomes large $\left( \alpha \gg 1 \right)$ the third location is
dominant and as $\alpha \rightarrow 1$ there is an even spread, as
in the standard microdiversity scenario. Using these parameters,
eight scenarios are given in Table \ref{table:mmse_zf_scenarios} for
the case of two single antenna users.
In Fig. \ref{fig:macro_capacity:s1-s5}, we explore the capacity of
scenarios S1-S4 where $n_R=3$. The capacity result in
\dref{eq:macro_capacity:dual_user} agrees with the simulations for
all scenarios, thus verifying the exact analysis. Furthermore, the
approximation in \dref{eq:macro_capacity:general:4} is shown to be
extremely accurate. All capacity results are extremely similar
except for S1, where both sources have their dominant path at the
first receive antenna. Here, the system is largely overloaded (two
strong signals at a single antenna) and the performance is lower.
The similarity of S3 and S4 is interesting as they represent very
different systems. In S3, the two users are essentially separate
with the dominant channels being at different antennas. In S4, both
users have power equally spread over all antennas so the users are
sharing all available channels. Figure
\ref{fig:macro_capacity:s6-s10} follows the same pattern with S6
(the overloaded case) being lower and the other scenarios almost
equivalent. In Fig. \ref{fig:macro_capacity:s6-s10}, the overall
capacity level is reduced in comparison to Fig.
\ref{fig:macro_capacity:s1-s5} as $\varsigma=10$ implies a weaker
second source.\\
Figures \ref{fig:macro_capacity:3user} and
\ref{fig:macro_capacity:6user} show results for a random drop
scenario with $M=n_R=3, W=N=3$ and $n_R=6, M=3, W=N=6$ respectively.
In each random drop, uniform random locations are created for the
users and lognormal shadow fading and path loss are considered where
$\sigma_{SF}=8$dB (standard deviation of shadow fading) and
$\gamma=3.5$ (path loss exponent). The transmit power of the sources
is scaled so that all locations in the coverage area have a maximum
received SNR greater than $3$dB, at least $95\%$ of the time. The
maximum SNR is taken over the 3 BSs. Hence, each drop produces a
different $\bm{P}$ matrix and independent channels are assumed.
The excellent agreement between the approximation in
\dref{eq:macro_capacity:general:4} and the simulations in both Fig.
\ref{fig:macro_capacity:3user} and Fig.
\ref{fig:macro_capacity:6user} is encouraging as this demonstrates
the accuracy of \dref{eq:macro_capacity:general:4} over different
system sizes as well as over completely different sets of channel
powers. Note that at high SNR, Fig. \ref{fig:macro_capacity:6user}
gives much higher capacity values than Fig.
\ref{fig:macro_capacity:3user} since there are 6 receive antenna
rather than 3. In this high SNR region, the $\bar{\gamma}^{N}$ term
in \dref{eq:macro_capacity:capacity_bound31} dominates and capacity
can be approximated by $\log_2 \left(1 + \Perm \left(\bm{P}
\right)\bar{\gamma}^{N}\right)$. With $n_R=N=3$ there are 6 cross
products in $\Perm \left(\bm{P} \right)$ whereas with $n_R=N=6$
there are 720 cross products. Hence, the bound clearly demonstrates
the benefits of increased antenna numbers. In practice there is a
trade-off between the costs of increased collaboration between
possibly distant BSs and the resulting increase in system capacity.
In Figs. \ref{fig:macro_capacity:s1-s5} -
\ref{fig:macro_capacity:6user}, at low SNR the capacity is
controlled by $P_T$. Hence, since $\rho=P_T / \sigma^2$, all four
drops have similar performance at low SNR and diverge at higher SNR
where the channel profiles affect the results.
The upper bound and associated approximations are shown in Figs.
\ref{fig:macro_capacity:user2_S3} and
\ref{fig:macro_capacity:random_scenario} both for a two user
scenario (S3) and a random drop. In Fig.
\ref{fig:macro_capacity:user2_S3}, the upper bound is shown for
scenario S3 as well as the high and low SNR approximations. The
results clearly show the loss in accuracy resulting from the use of
the simple Jensen bound. However, the bound is quite reasonable over
the whole SNR range. The low SNR approximations are quite reasonable
below 0 dB and the high SNR version is as accurate as the bound
above 15 dB. In Fig. \ref{fig:macro_capacity:random_scenario},
similar results are shown for a random drop with $M=n_R=6, W=N=6$.
Here, similar patterns are observed but the low and high SNR
approximations become reasonable at more widely spread SNR values.
For example, the low SNR results are accurate below 0 dB and the
high SNR results are poor until around 30 dB. In contrast, the upper
bound is reasonable throughout. Hence, although there is a
noticeable capacity error at high SNR, the cross-product
coefficients in \dref{eq:macro_capacity:capacity_bound:2by2} and
\dref{eq:macro_capacity:capacity_bound:3by3} are seen to explain the
large majority of the ergodic capacity behavior.

\begin{table}[Parameters for Figures]
    \caption{Parameters for Figures \ref{fig:macro_capacity:s1-s5} and \ref{fig:macro_capacity:s6-s10}}
    \centering
    \begin{tabular}{ c | l | l | c }
    \hline \hline
    & \multicolumn{2}{|c|}{Decay Parameter} & \\ \cline{2-3}
    Sc. No.  & User 1 & User 2 &  $\varsigma$  \\ \hline
    S1 &  $\alpha=0.1$ & $\alpha=0.1$ & 1\\
    S2 &  $\alpha=0.1$ & $\alpha=1$ & 1 \\
    S3 &  $\alpha=0.1$ & $\alpha=10$ & 1 \\
    S4 &  $\alpha=1$ & $\alpha=1$ & 1 \\ \hline
    S5 & $\alpha=0.1$ & $\alpha=0.1$ & 10  \\
    S6  & $\alpha=0.1$ & $\alpha=1$ & 10\\
    S7  & $\alpha=0.1$ & $\alpha=10$ & 10\\
    S8 & $\alpha=1$ & $\alpha=0.1$ & 10\\
    \hline
    \end{tabular}
\label{table:mmse_zf_scenarios} 
\end{table}

\begin{figure}[h]
\centerline{\includegraphics*[scale=0.65]{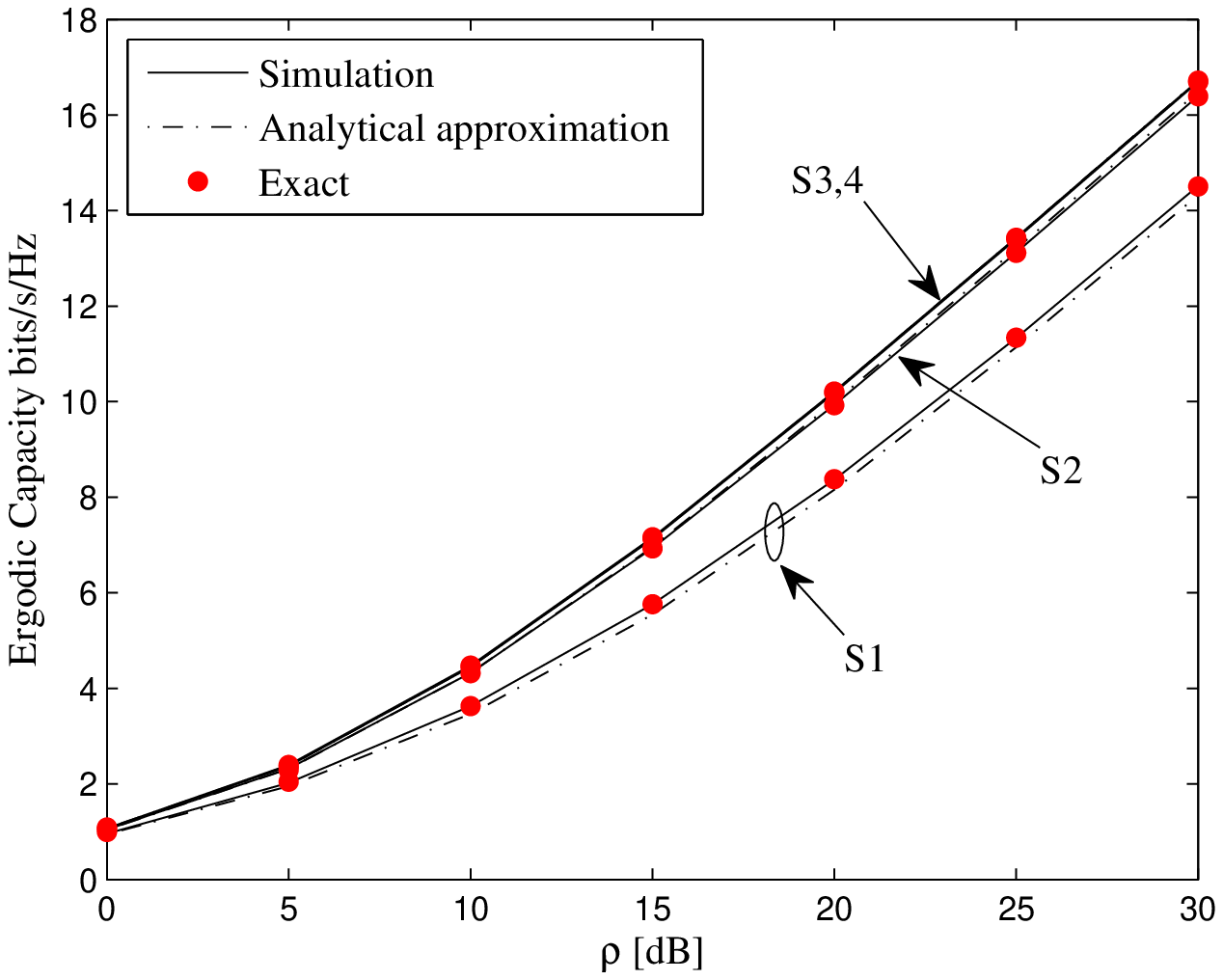}}
\caption{Exact, approximated and simulated ergodic sum capacity in
flat Rayleigh fading for scenarios S1-S4 with parameters: $n_R=3$,
$N=W=2$ and $\varsigma=1$.} \label{fig:macro_capacity:s1-s5}
\end{figure}

\begin{figure}[h]
\centerline{\includegraphics*[scale=0.65]{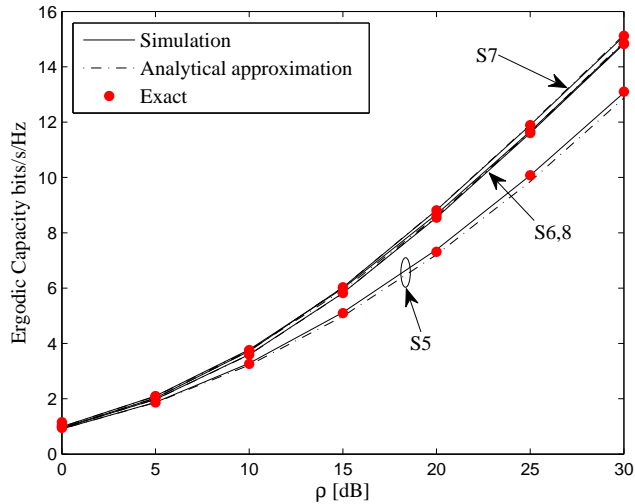}}
\caption{Exact, approximated and simulated ergodic sum capacity in
flat Rayleigh fading for scenarios S5-S8 with parameters: $n_R=3$,
$N=W=2$ and $\varsigma=10$.} \label{fig:macro_capacity:s6-s10}
\end{figure}
\begin{figure}[h]
\centerline{\includegraphics*[scale=0.65]{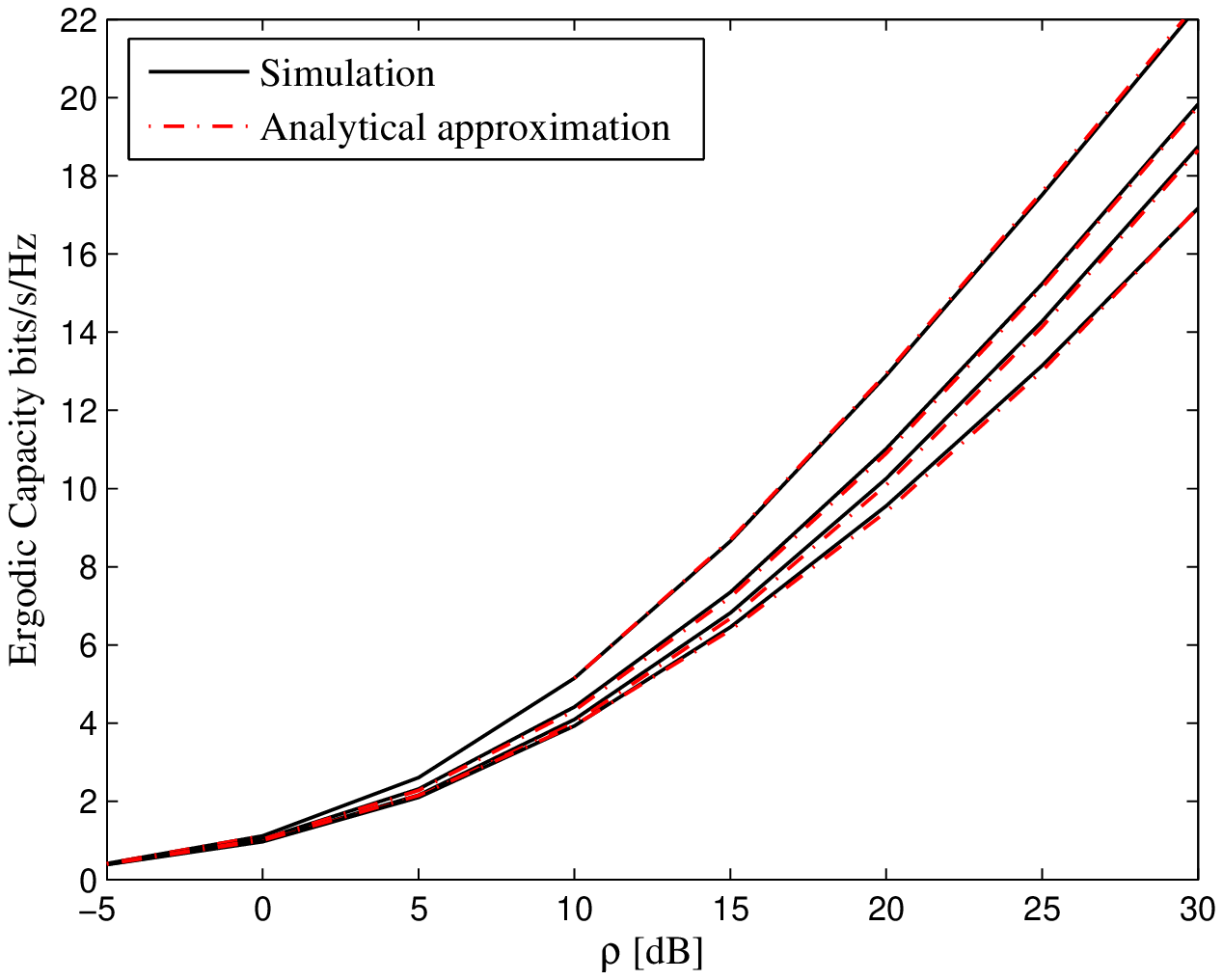}}
\caption{Approximated and simulated ergodic sum capacity in flat
Rayleigh fading for $M=n_R=3$, $W=N=3$ and four random drops.}
\label{fig:macro_capacity:3user}
\end{figure}

\begin{figure}[h]
\centerline{\includegraphics*[scale=0.65]{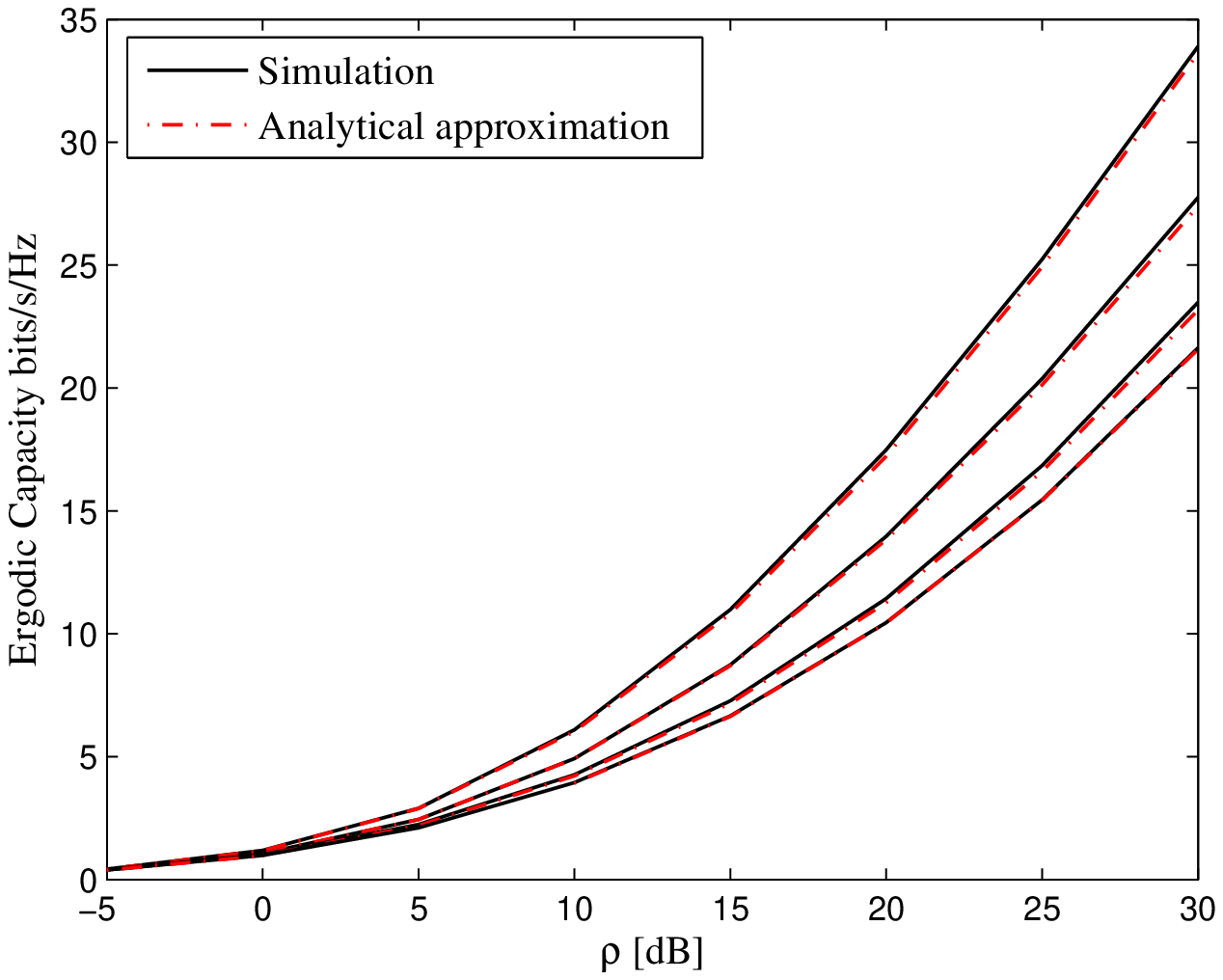}}
\caption{Approximated and simulated ergodic sum capacity in flat
Rayleigh fading for $n_R=6$, $M=3$, $W=N=6$ and four random drops.}
\label{fig:macro_capacity:6user}
\end{figure}

\begin{figure}[h]
\centerline{\includegraphics*[scale=0.65]{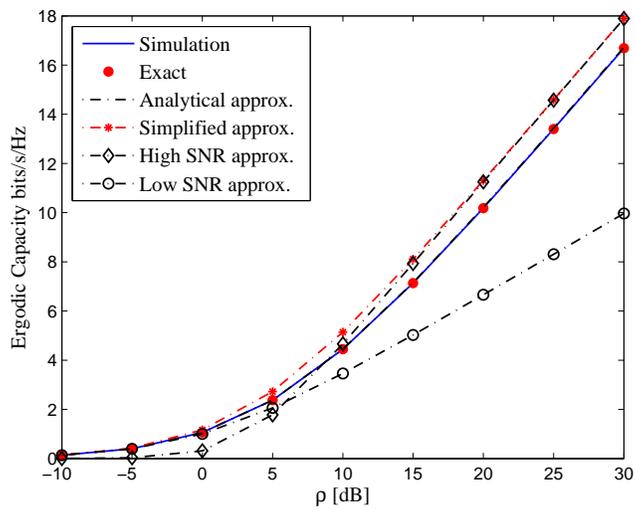}}
\caption{Ergodic sum capacity in flat Rayleigh fading for scenario
S3 with parameters: $M=n_R=3$, $W=N=2$ and $\varsigma=1$.}
\label{fig:macro_capacity:user2_S3}
\end{figure}

\begin{figure}[h]
\centerline{\includegraphics*[scale=0.65]{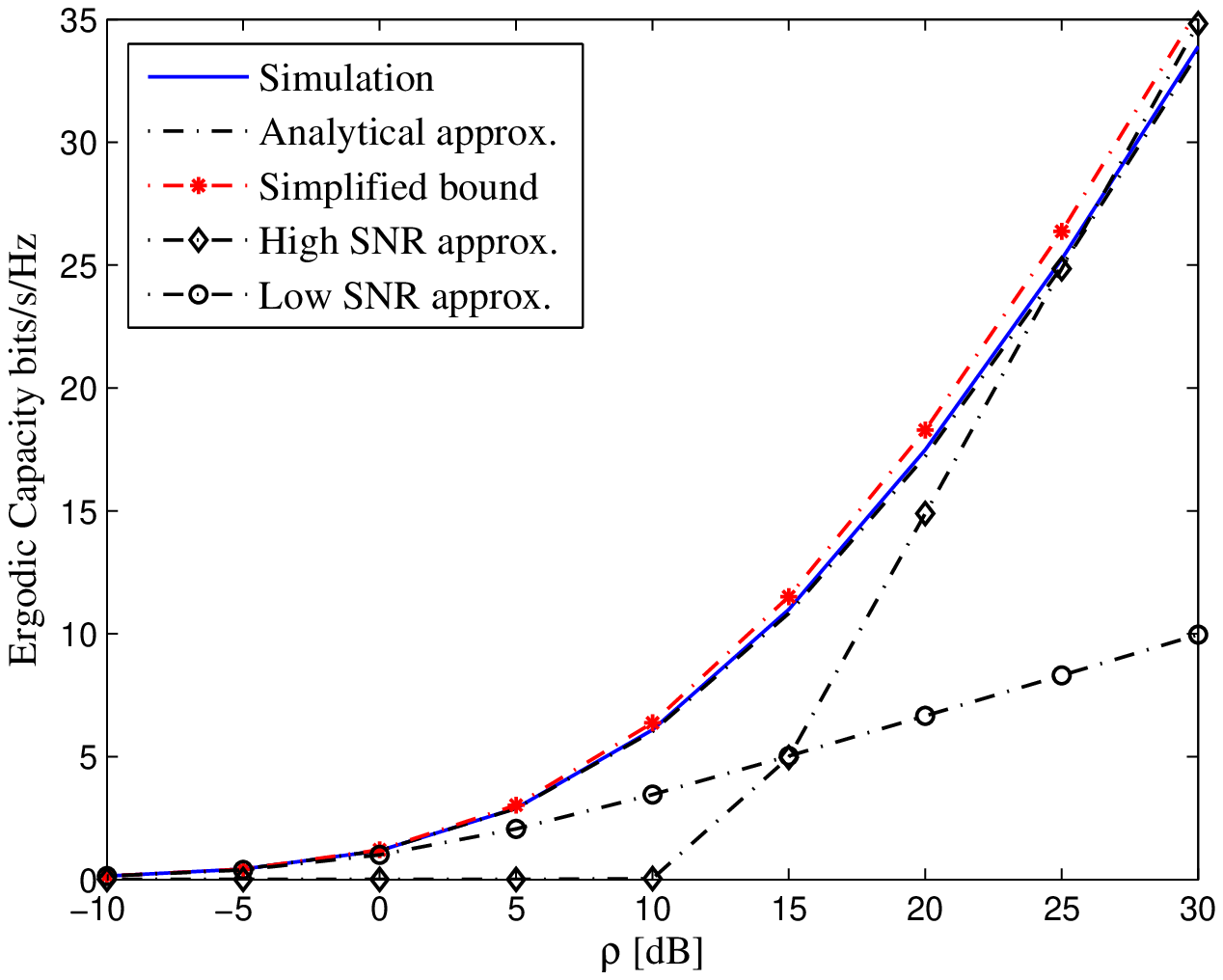}}
\caption{Ergodic sum capacity in flat Rayleigh fading for a random
drop with parameters: $n_R=6$, $M=3$ and $W=N=6$.}
\label{fig:macro_capacity:random_scenario}
\end{figure}

\section{Conclusion}\label{sec:macro_capacity:conclusion}
In this paper, we have studied the ergodic sum capacity of a
Rayleigh fading macrodiversity MIMO-MAC. The results obtained are
shown to be valid for both independent channels and correlated
channels, which may occur when some of the distributed
transmit/receive locations have closely spaced antennas. In
particular, we derive exact results for the two-source scenario and
approximate results for the general case. The approximations have a
simple form and are shown to be very accurate over a wide range of
channel powers. In addition, a simple upper bound is presented which
demonstrates the importance of various channel power cross products
in determining capacity.
\appendices
\section{Derivation of $I_b$ }\label{app:macro_capacity:A}
From \dref{eq:macro_capacity:I_b2}, $I_b$ can be written as
\begin{align}
I_b &= \left. \frac{\partial \tilde{I}_b}{\partial \theta_1}
\right|_{\theta_1=0}, \label{eq:macro_capacity:I_b3}
\end{align}
where
\begin{equation}
\tilde{I}_{b} = -\!\! \int_0^\infty \!\!\!\! \int_0^\infty \!\!
\frac{e^{-\sigma^2t - \sigma^2 \theta_2}}{ \prod_{i=1 }^{n_R} \left(
1 + t P_{i2} + \theta_1 P_{i1} P_{i2} + \theta_2 P_{i1} \right) }
d\theta_2 dt. \label{eq:macro_capacity:I_b'1}
\end{equation}
From \dref{eq:macro_capacity:I_b'1}, $L_{b}$ becomes
\begin{align}
L_{b} = \int_0^\infty \!\!\!\! \int_0^\infty \!\!\!\!
\frac{e^{-\sigma^2t - \sigma^2 \theta_2}}{ \prod_{i=1 }^{n_R} \left(
\theta_1 + \frac{\theta_2}{P_{i2}} + \frac{t}{P_{i1}} +
\frac{1}{P_{i1}P_{i2}} \right) } d\theta_2 dt.
\label{eq:macro_capacity:I_b'2}
\end{align}
Defining
\begin{align}
L_{b} &=  - \left| \bm{P}_1 \bm{P}_2 \right| \tilde{I}_{b},
\label{eq:macro_capacity:I_b''1}
\end{align}
we use a partial fraction expansion in $\theta_1$ to give
\begin{align}
L_{b} =  \sum_{i=1}^{n_R} \int_0^\infty \!\!\!\! \int_0^\infty \!
 \frac{A_i \left(\theta_2,t \right) e^{-\sigma^2 t - \sigma^2 \theta_2}}{\left(
\theta_1 + \frac{\theta_2}{P_{i2}} + \frac{t}{P_{i1}} +
\frac{1}{P_{i1}P_{i2}} \right) } d\theta_2 dt,
\label{eq:macro_capacity:I_b''2}
\end{align}
\noindent where
\begin{subequations}
\begin{align}
A_i  \left(\theta_2,t \right) &= \frac{1}{\prod_{k \neq i }^{n_R}
\left( \alpha_{ik}\theta_2 + \beta_{ik}t + \gamma_{ik}
\right)}\label{eq:macro_capacity:eta1} \\
\alpha_{ik} &= \frac{1}{P_{k2}} - \frac{1}{P_{i2}}\\
\beta_{ik} &= \frac{1}{P_{k1}} - \frac{1}{P_{i1}} \\
\gamma_{ik} &= R_k - R_i \\
R_i &= \frac{1}{P_{i1}P_{i2}}.
\end{align}
\end{subequations}
To compute \dref{eq:macro_capacity:I_b''2}, the following
substitutions are employed
\begin{subequations}
\begin{align}
u= \sigma^2 t + \sigma^2 \theta_2 \\
v_i = \frac{t}{P_{i1}} +  \frac{\theta_2}{P_{i2}}.
\label{eq:macro_capacity:transformation1}
\end{align}
\end{subequations}
The Jacobian of the transformation in
\dref{eq:macro_capacity:transformation1} can be calculated as
\begin{align}
J_i = \sigma^2 \left( \frac{1}{P_{i2}} - \frac{1}{P_{i1}} \right).
\label{eq:macro_capacity:jacobian}
\end{align}
Substituting \dref{eq:macro_capacity:transformation1} and
\dref{eq:macro_capacity:jacobian} in \dref{eq:macro_capacity:I_b''2}
gives
\begin{align}
L_{b} =  \sum_{i=1}^{n_R} \int_0^\infty \!\!\!\!
\int_{\frac{u}{P_{i1}\sigma^2}}^{\frac{u}{P_{i2}\sigma^2}} \frac{A_i
\left(u,v_i \right) e^{-u}}{J_i \left(  v_i + \theta_1 + R_i
\right)} dv_i du, \label{eq:macro_capacity:I_b''3}
\end{align}
\noindent where
\begin{subequations}
\begin{align}
A_i  \left(u,v_i \right) &= \frac{1}{\prod_{k \neq i }^{n_R} \left(
a_{ik}v_i + b_{ik}u + \gamma_{ik}
\right)}\label{eq:macro_capacity:eta2}\\
a_{ik} &= \frac{\sigma^2}{J_i}\left( \alpha_{ik} - \beta_{ik} \right) \\
b_{ik} &= \frac{1}{J_i}\left( \frac{\beta_{ik}}{P_{i2}} -
\frac{\alpha_{ik}}{P_{i1}} \right).
\end{align}
\end{subequations}
The term $A_i  \left(u,v_i \right)$ in \dref{eq:macro_capacity:eta2}
can be written as a summation using partial fractions, to give
\begin{align}
A_i  \left(u,v_i \right) &= \sum_{k \neq i }^{n_R}
\frac{B_{ik}\left(u\right)}{v_i + q_{ik} u + r_{ik}},
\label{eq:macro_capacity:eta3}
\end{align}
\noindent where
\begin{subequations}
\begin{align}
B_{ik} \left(u\right) &= \frac{\left( a_{ik} \right
)^{n_R-3}}{\prod_{l
\neq i,k }^{n_R} \left( c_{ikl} u +  d_{ikl}\right)} \label{eq:macro_capacity:Bu} \\
c_{ikl} &= b_{il} a_{ik} - a_{il} b_{ik} \\
d_{ikl} &=  a_{ik} \gamma_{il}  -
\gamma_{ik} a_{il}\\
q_{ik} &= \frac{b_{ik}}{a_{ik}}\\
r_{ik} &= \frac{\gamma_{ik}}{a_{ik}}.
\end{align}
\end{subequations}
Substituting \dref{eq:macro_capacity:eta3} in
\dref{eq:macro_capacity:I_b''3} and simplifying gives
\begin{align}
\!\!\!\!\!L_{b} =  \sum_{i=1}^{n_R} &\sum_{k \neq i }^{n_R}
\int_0^\infty
\int_{\frac{u}{P_{i1}\sigma^2}}^{\frac{u}{P_{i2}\sigma^2}}
\frac{B_{ik}  \left(u\right) e^{-u}}{J_i}
\left. \frac{dv_i du}{\left(  v_i + \theta_1 + R_i
\right)\left( v_i + q_{ik} u + r_{ik} \right)} \right. .
\label{eq:macro_capacity:I_b''4}
\end{align}
First, we integrate over $v_i$ in \dref{eq:macro_capacity:I_b''4} to
give
\begin{align}
\!\!\!\!\!\!\!L_{b} =  \sum_{i=1}^{n_R} &\sum_{k \neq i }^{n_R}
\int_0^\infty
\frac{C_{ik}  \left(u, \theta_1 \right) e^{-u}}{J_i}
 \ln \left[ \frac{\left(\frac{u}{P_{i2}\sigma^2} + \theta_1 +
R_i \right)\!\!\left( \lambda_{ik}u + r_{ik}
\right)}{\left(\frac{u}{P_{i1}\sigma^2} + \theta_1 + R_i
\right)\!\!\left( \mu_{ik}u + r_{ik} \right)} \right ] du,
\label{eq:macro_capacity:I_b''5}
\end{align}

\noindent where
\begin{subequations}
\begin{align}
C_{ik}  \left(u, \theta_1 \right) &= \frac{B_{ik}
\left(u\right)}{q_{ik}u + r_{ik} - \theta_1 - R_i} \\
\lambda_{ik} &= \frac{1}{P_{i1}\sigma^2} + q_{ik} \\
\mu_{ik} &= \frac{1}{P_{i2}\sigma^2} + q_{ik}.
\label{eq:macro_capacity:eta01}
\end{align}
\end{subequations}
Let
\begin{subequations}
\begin{align}
D_{ik} \left( u, \theta_1 \right) &= \ln \left[
\frac{\left(\frac{u}{P_{i2}\sigma^2} + \theta_1 + R_i
\right)\!\!\left( \lambda_{ik}u + r_{ik}
\right)}{\left(\frac{u}{P_{i1}\sigma^2} + \theta_1 + R_i
\right)\!\!\left( \mu_{ik}u + r_{ik} \right)} \right],
\label{eq:macro_capacity:omega1}
\end{align}
\end{subequations}
then $B_{ik} \left(u\right)$ in \dref{eq:macro_capacity:Bu} can be
rewritten as the summation
\begin{align}
B_{ik} \left(u\right) = \sum_{l \neq i,k }^{n_R}
\frac{\xi_{ikl}}{c_{ikl} u +  d_{ikl}},
\label{eq:macro_capacity:eta'1}
\end{align}
\noindent where
\begin{align}
\xi_{ikl} = \frac{\left(a_{ik} c_{ikl}\right)^{n_R-3}}{\prod_{z \neq
i,k,l }^{n_R} \left( d_{ikz} c_{ikl} - c_{ikz} d_{ikl} \right)}.
\label{eq:macro_capacity:xi1}
\end{align}
Substituting \dref{eq:macro_capacity:eta'1} and
\dref{eq:macro_capacity:eta01} in \dref{eq:macro_capacity:I_b''5}
gives
\begin{align}
L_{b} =  \sum_{i=1}^{n_R} &\sum_{k \neq i }^{n_R} \sum_{l \neq i,k
}^{n_R} \int_0^\infty D_{ik} \left( u, \theta_1 \right) \frac{\xi_{ikl}}{J_i}
\frac{du}{\left( c_{ikl} u + d_{ikl} \right) \left( q_{ik}u
+ r_{ik} - \theta_1 - R_i \right)}.
\label{eq:macro_capacity:I_b''6}
\end{align}
Equation \dref{eq:macro_capacity:I_b''6} can be further simplified
to give
\begin{align}
L_{b} &= \sum_{i=1}^{n_R} \sum_{k \neq i }^{n_R} \sum_{l \neq i,k
}^{n_R} \frac{\xi_{ikl}\left( M_{b_{ikl}} - N_{b_{ikl}}
\right)}{J_i} , \label{eq:macro_capacity:I_b''7}
\end{align}
\noindent where
\begin{align}
M_{b_{ikl}} &= \int_0^\infty \frac{D_{ik}\left( u, \theta_1
\right)}{ f_1\left(\theta_1 \right)}  \frac{du}{\left( u +
\varepsilon_{ikl} \right)},  \label{eq:macro_capacity:I1_bikl1}
\end{align}
\begin{align}
N_{b_{ikl}} &= \int_0^\infty \frac{D_{ik}\left( u, \theta_1
\right)}{ f_1\left(\theta_1 \right)}  \frac{du}{\left( u +
f_2\left(\theta_1 \right) \right)},
\label{eq:macro_capacity:I2_bikl1}
\end{align}
\noindent and $\varepsilon_{ikl} = d_{ikl} / c_{ikl}$. Next, we
introduce the following linear functions of $\theta_1$:
\begin{align}
f_1\left(\theta_1 \right) &= n_{ikl} - c_{ikl} \theta_1 \label{eq:macro_capacity:theta_1_functions1}\\
f_2\left(\theta_1 \right) &= m_{ikl} - \frac{\theta_1}{q_{ik}},
\label{eq:macro_capacity:theta_1_functions2}
\end{align}
\noindent where
\begin{align}
n_{ikl} &= r_{ik} c_{ikl} - d_{ikl} q_{ik} -
\frac{c_{ikl}}{R_i} \\
m_{ikl} &= \frac{\gamma_{ik}}{b_{ik}} - \frac{1}{q_{ik} R_i}.
\end{align}
Next, we can differentiate $M_{b_{ikl}}$ and $N_{b_{ikl}}$ and
integrate over $u$ to give the final result along with
\dref{eq:macro_capacity:I_b3} and \dref{eq:macro_capacity:I_b''1}.
Hence, from \dref{eq:macro_capacity:theta_1_functions1} and
\dref{eq:macro_capacity:omega1} we get
\dref{eq:macro_capacity:omega_f1_diff_set_theta1_zero}.
\begin{table*}[!t]
\small
\begin{align}
\frac{\partial}{\partial \theta_1}\left[ \frac{D_{ik}\left( u,
\theta_1 \right)}{f_1\left(\theta_1 \right)} \right]_{\theta_1=0}
\!=\! \frac{c_{ikl}}{n_{ikl}^2} &\ln \left[
\frac{\left(\frac{u}{P_{i2}\sigma^2} + R_i \right)\left(\lambda_{ik}
u + r_{ik} \right)}{\left(\frac{u}{P_{i1}\sigma^2} + R_i
\right)\left(\mu_{ik} u + r_{ik} \right)} \right] +
\frac{1}{n_{ikl}} \left[ \frac{P_{i2}\sigma^2}{\left(u +
\frac{\sigma^2}{P_{i1}} \right)}  - \frac{P_{i1}\sigma^2}{\left(u +
\frac{\sigma^2}{P_{i2}} \right)}\right]
\label{eq:macro_capacity:omega_f1_diff_set_theta1_zero}
\end{align}
\begin{align}
\tilde{M}_{b_{ikl}} &= \left. \frac{\partial M_{b_{ikl}}}{\partial
\theta_1} \right|_{\theta_1=0} = \int_0^\infty
\frac{\partial}{\partial \theta_1}\left[ \frac{D_{ik}\left( u,
\theta_1 \right)}{f_1\left(\theta_1 \right)} \right]_{\theta_1=0}
\frac{du}{\left( u + e_{ikl} \right)}
\label{eq:macro_capacity:I1_bikl2}
\end{align}
\begin{align}
\tilde{N}_{b_{ikl}} &= \left. \frac{\partial N_{b_{ikl}}}{\partial
\theta_1} \right|_{\theta_1=0} = \int_0^\infty
\frac{\partial}{\partial \theta_1}\left[ \frac{D_{ik}\left( u,
\theta_1 \right)}{f_1\left(\theta_1 \right)} \right]_{\theta_1=0}
\frac{du}{\left( u + m_{ikl} \right)} + \int_0^\infty \left[
\frac{D_{ik}\left( u, \theta_1 \right)}{f_1\left(\theta_1 \right)}
\right]_{\theta_1=0} \frac{1/q_{ik}}{\left( u + m_{ikl} \right)^2}
du \label{eq:macro_capacity:I2_bikl2}
\end{align}
\hrulefill
\begin{align}
\tilde{M}_{b_{ikl}}  &= \frac{c_{ikl}}{n_{ikl}^2} \left[
H_1\left(R_i, \varepsilon_{ikl}, \frac{1}{P_{i2}\sigma^2}\right) +
H_1\left(r_{ik} , \varepsilon_{ikl}, \lambda_{ik} \right) -
H_1\left( R_i, \varepsilon_{ikl}, \frac{1}{P_{i1}\sigma^2}\right) -
H_1\left(r_{ik} , \varepsilon_{ikl}, \mu_{ik}
\right) \right] \nonumber \\
&+ \frac{\varepsilon'_{ikl}}{n_{ikl}} \left[
e^{\frac{\sigma^2}{P_{i1}}} E_1 \left(\frac{\sigma^2}{P_{i1}}\right)
- e^{\varepsilon_{ikl}} E_1 \left( \varepsilon_{ikl} \right) \right]
- \frac{\varepsilon''_{ikl}}{n_{ikl}} \left[
e^{\frac{\sigma^2}{P_{i2}}} E_1 \left(\frac{\sigma^2}{P_{i2}}\right)
- e^{\varepsilon_{ikl}} E_1 \left( \varepsilon_{ikl} \right)
\right], \label{eq:macro_capacity:I1_bikl3}
\end{align}
\par
\hrulefill
\begin{align}
\tilde{N}_{b_{ikl}}  &= \frac{c_{ikl} }{n_{ikl}^2 q_{ik}} \left[
H_2\left(R_i, m_{ikl}, \frac{1}{P_{i2}\sigma^2}\right) + H_2\left(
r_{ik} , m_{ikl}, \lambda_{ik} \right) - H_2\left( R_i, m_{ikl},
\frac{1}{P_{i1}\sigma^2}\right) - H_2\left(r_{ik} , m_{ikl},
\mu_{ik}
\right) \right]  \nonumber \\
&+ \frac{c_{ikl}}{n_{ikl}^2} \left[ H_1\left(R_i, m_{ikl},
\frac{1}{P_{i2}\sigma^2}\right) + H_1\left(r_{ik} , m_{ikl},
\lambda_{ik} \right) - H_1\left( R_i, m_{ikl},
\frac{1}{P_{i1}\sigma^2}\right) - H_1\left(r_{ik} , m_{ikl},
\mu_{ik}
\right) \right] \nonumber \\
&+ \frac{m'_{ikl}}{n_{ikl}} \left[ e^{\frac{\sigma^2}{P_{i1}}} E_1
\left(\frac{\sigma^2}{P_{i1}}\right) - e^{m_{ikl}} E_1 \left(
m_{ikl} \right) \right] - \frac{m''_{ikl}}{n_{ikl}} \left[
e^{\frac{\sigma^2}{P_{i2}}} E_1 \left(\frac{\sigma^2}{P_{i2}}\right)
- e^{m_{ikl}} E_1 \left( m_{ikl} \right) \right],
\label{eq:macro_capacity:I2_bikl3}
\end{align}
\par
\par
\hrulefill
\end{table*}
Substituting \dref{eq:macro_capacity:omega_f1_diff_set_theta1_zero}
in \dref{eq:macro_capacity:I1_bikl1} and
\dref{eq:macro_capacity:I2_bikl1} we get
\dref{eq:macro_capacity:I1_bikl2} and
\dref{eq:macro_capacity:I2_bikl2}. \dref{eq:macro_capacity:I1_bikl2}
and \dref{eq:macro_capacity:I2_bikl2} can be solved in closed form
to give \dref{eq:macro_capacity:I1_bikl3} and
\dref{eq:macro_capacity:I2_bikl3}, where we have used the two
integrals defined as follows
\begin{align}
H_1 \left(a,b,c \right) &=  \int_0^\infty \frac{e^{-t} \ln \left(c t
+ a \right)}{t + b }dt \label{eq:macro_capacity:H1} \\
H_2 \left(a,b,c \right) &=  \int_0^\infty \frac{e^{-t} \ln \left(c t
+ a \right)}{\left( t + b \right)^2 }dt \label{eq:macro_capacity:H2}
\end{align}
\noindent and the constants are given by
\begin{align}
\varepsilon'_{ikl} = \frac{1}{\left( \varepsilon_{ikl} -
\frac{\sigma^2}{P_{i1}} \right)} , \quad  \varepsilon''_{ikl} &=
\frac{1}{\left( \varepsilon_{ikl} -
\frac{\sigma^2}{P_{i2}} \right)}, \nonumber \\
m'_{ikl} = \frac{1}{\left( m_{ikl} - \frac{\sigma^2}{P_{i1}}
\right)} , \quad m''_{ikl} &= \frac{1}{\left( m_{ikl} -
\frac{\sigma^2}{P_{i2}} \right)}. \nonumber
\end{align}
Both $H_1$ and $H_2$ can be solved in closed form as
\begin{equation}
H_1 \left(a,b,c \right) =  e^b \left[ E_1 \left(b\right) \ln c +
D_1\left(\frac{a}{c}-b, b \right) \right], \nonumber
\label{eq:macro_capacity:H1_answer}
\end{equation}
\begin{equation}
H_2 \left(a,b,c \right) =  \ln c \left[ \frac{1}{b} - e^b E_1
\left(b\right) \right] - 2 e^b D_1\left(\frac{a}{c}-b, b \right) +
\frac{1}{\left( \frac{a}{c}-b \right)} \left[e^b E_1 \left(b\right)
- e^{\frac{a}{c}} E_1 \left(\frac{a}{c}\right) \right], \nonumber
\label{eq:macro_capacity:H1_answer}
\end{equation}
\noindent where $D_1(a,b)$ is defined by
\begin{equation}
D_1(a,b) = \int_b^\infty \frac{e^{-t} \ln \left(t + a \right)}{t}dt,
\hspace{5mm} \mbox{for} \hspace{5mm} b\neq 0. \nonumber
\label{eq:macro_capacity:D1}
\end{equation}
\section{Calculation of $E \left\{ \left|\sigma^2 \bm{I} + \tilde{\bm{H}}_k^H \tilde{\bm{H}}_k\right|
\right \}$}\label{app:macro_capacity:numerator_expectation}
Let $\lambda_1, \lambda_2, \dots , \lambda_{k-1}$ be the ordered
eigenvalues of $\tilde{\bm{H}}_k^H \tilde{\bm{H}}_k$. Since $n_R
\geq \left(k-1 \right)$, all eigenvalues are non zero. Then,
\begin{align}
E \left\{ \left|\sigma^2 \bm{I} + \tilde{\bm{H}}_k^H
\tilde{\bm{H}}_k\right| \right \} &= E \left\{ \prod_{i=1}^{k-1}
\left( \sigma^2 + \lambda_i \right) \right
\}  \nonumber \\
\!&=\! E \!\left\{ \sum_{i=0}^{k-1} \Tr_{i} \!\left(
\tilde{\bm{H}}_k^H \tilde{\bm{H}}_k \right)\left(\sigma^2
\right)^{k-i-1} \!\! \right \},
\label{eq:macro_capacity:mmse_expectation:1}
\end{align}
where \dref{eq:macro_capacity:mmse_expectation:1} is from
\dref{identity:macro_capacity:esf} and Lemma
\ref{lemma:macro_capacity:esf_trace_identity}. Therefore, the
building block of this expectation is $E \left\{\Tr_{i} \left(
\tilde{\bm{H}}_k^H \tilde{\bm{H}}_k \right) \right \}$. From Lemma
\ref{lemma:macro_capacity:esf_trace_identity}
\begin{align}
\Tr_{i} \left( \tilde{\bm{H}}_k^H \tilde{\bm{H}}_k \right) &=
\sum_{\sigma} \left| \left( \tilde{\bm{H}}_k^H \tilde{\bm{H}}_k
\right)_{\sigma_{i,k-1}} \right|.
\label{eq:macro_capacity:mmse_expectation:2}
\end{align}
Therefore, from Lemma
\ref{lemma:macro_capacity:expected_square_determinant},
\begin{equation}
E \left\{ \Tr_{i} \left( \tilde{\bm{H}}_k^H \tilde{\bm{H}}_k \right)
\right \} = \sum_{\sigma} \Perm \left( \left(\bm{Q}_k
\right)^{\sigma_{i,k-1}} \right), \nonumber
\label{eq:macro_capacity:mmse_expectation:3}
\end{equation}
\noindent where the $n_R \times \left( k-1 \right)$ matrix,
$\bm{Q}_k$, is given by
\begin{equation}
E \left\{ \tilde{\bm{H}}_k \circ \tilde{\bm{H}}_k \right\} =
\bm{Q}_k. \label{eq:macro_capacity:Qk1}
\end{equation}
Note that summation in \dref{eq:macro_capacity:mmse_expectation:3}
has $\binom {k-1}{i}$ terms. Then, the final expression becomes
\begin{align}
E \!\left\{ \left|\sigma^2 \bm{I} + \tilde{\bm{H}}_k^H
\tilde{\bm{H}}_k\right| \right \} \!&=\!\! \sum_{i=0}^{k-1}\!
\sum_{\sigma} \Perm \left( \left(\bm{Q}_k \right)^{\sigma_{i,k-1}}
\right)\! \left(\sigma^2 \right)^{k-i-1}.
\label{eq:macro_capacity:mmse_expectation:final_expression}
\end{align}
\section{Calculation of $\left| \bm{\Sigma}_k \right| E
\left\{ \left| \sigma^2 \bm{I} + \tilde{\bm{H}}_k^H \bm{\Sigma}_k^{-1} \tilde{\bm{H}}_k\right| \right\}$ } \label{app:macro_capacity:denom_expectation}
A simple extension of \dref{eq:macro_capacity:Cj_numerator:expectation}
allows the expectation in the denominator of \dref{eq:macro_capacity:general:ICj2}
to be calculated as
\begin{align}
E \left\{ \left|\sigma^2 \bm{I} + \tilde{\bm{H}}_k^H
\bm{\Sigma}_k^{-1} \tilde{\bm{H}}_k\right| \right \} \!\! &=\!\!
\sum_{i=0}^{k-1} \psi_{ki} \left(t\right) \left(\sigma^2
\right)^{k-i-1},
\label{eq:app:macro_capacity:ICj2_denominator:expectation}
\end{align}
\noindent where
\begin{align}
\psi_{ki} \left(t\right) &= \sum_{\sigma} \Perm \left(
\left(\bm{\Sigma}_k^{-1} \bm{Q}_k \right)^{\sigma_{i,k-1}} \right),
\label{eq:app:macro_capacity:ICj2_denominator:expectation:psi1}
\end{align}
and from \dref{eq:macro_capacity:axiom1}
\begin{equation}
\psi_{k0} \left(t\right) = 1. \nonumber
\label{eq:app:macro_capacity:ICj2_denominator:expectation:psi_zero}
\end{equation}
The term in
\dref{eq:app:macro_capacity:ICj2_denominator:expectation:psi1} can
be simplified using \dref{eq:macro_capacity:axiom2} to obtain
\begin{align}
\psi_{ki} \left(t\right) &=  \sum_\sigma \frac{\Perm \left( \left(
\bm{Q}_k
\right)_{\sigma_{i,n_R}}^{\{k-1\}}\right)}{\left|\left(\bm{\Sigma}_k
\right)_{\sigma_{i,n_R}}\right|}.
\label{eq:app:macro_capacity:ICj2_denominator:expectation:psi2}
\end{align}
Then,
\begin{align}
\left|\bm{\Sigma}_k \right| E \left\{ \left|\sigma^2 \bm{I} +
\tilde{\bm{H}}_k^H \bm{\Sigma}_k^{-1} \tilde{\bm{H}}_k \right|
\right \} \!\! &=\!\! \sum_{i=0}^{k-1} \xi_{ki} \left(t\right)
\left(\sigma^2 \right)^{k-i-1},
\label{eq:app:macro_capacity:ICj2_denominator:expectation1}
\end{align}
\noindent where $\xi_{ki} \left(t\right) =\left|\bm{\Sigma}_k\right|
\psi_{ki} \left(t\right)$. From
\dref{eq:app:macro_capacity:ICj2_denominator:expectation:psi2}, we
obtain
\begin{align}
\xi_{ki} \left(t\right) &= \sum_\sigma \left|\left( \bm{\Sigma}_k
\right)_{\bar{\sigma}_{n_R-i,n_R}} \right| \Perm \left( \left(
\bm{Q}_k \right)_{\sigma_{i,n_R}}^{\{k-1\}}\right),
\label{eq:app:macro_capacity:ICj2_denominator:expectation:xi1}
\end{align}
\noindent where $\bar{\sigma}_{n_R-i,n_R}$ is the compliment of
$\sigma_{i,n_R}$. Therefore, it is apparent that $\xi_{ki}
\left(t\right)$ is a polynomial of degree $n_R-i$. Clearly
$\left|\bm{\Sigma}_k\right| E \left\{ \left|\sigma^2 \bm{I} +
\tilde{\bm{H}}_k^H \bm{\Sigma}_k^{-1} \tilde{\bm{H}}_k\right| \right
\}$ is a polynomial of degree $n_R$, since
$\xi_{k0}\left(t\right)=\left|\bm{\Sigma}_k\right|$ is the highest
degree polynomial term in $t$ in
\dref{eq:app:macro_capacity:ICj2_denominator:expectation1}. Then,
\begin{equation}
\left|\left(\bm{\Sigma}_k \right)_{\bar{\sigma}_{n_R-i,n_R}} \right|
= \sum_{l=0}^{n_R-i} \left(\frac{t}{\sigma^2}\right)^l \Tr_l \left(
\left(\bm{P}_k\right)_{\bar{\sigma}_{n_R-i,n_R}} \right).
\label{eq:app:macro_capacity:ICj2_denominator:expectation:anonymous1}
\end{equation}
Hence, applying
\dref{eq:app:macro_capacity:ICj2_denominator:expectation:anonymous1}
in \dref{eq:app:macro_capacity:ICj2_denominator:expectation:xi1},
\begin{equation}
\xi_{ki} \left(t\right) = \sum_\sigma \sum_{l=0}^{n_R-i}
\left(\frac{t}{\sigma^2}\right)^l \Tr_l \left(
\left(\bm{P}_k\right)_{\bar{\sigma}_{n_R-i,n_R}} \right)
 \Perm \left( \left( \bm{Q}_k
\right)_{\sigma_{i,n_R}}^{\{k-1\}}\right), \nonumber
\label{eq:app:macro_capacity:ICj2_denominator:expectation:xi2}
\end{equation}
and $\xi_{ki} \left(t\right)$ becomes
\begin{align}
\xi_{ki} \left(t\right) &= \sum_{l=0}^{n_R-i}
\left(\frac{t}{\sigma^2}\right)^l \hat{\varphi}_{kli} \label{eq:app:macro_capacity:ICj2_denominator:expectation:xi30}\\
&= \sum_{l=0}^{n_R} \left(\frac{t}{\sigma^2}\right)^l
\hat{\varphi}_{kli},
\label{eq:app:macro_capacity:ICj2_denominator:expectation:xi3}
\end{align}
\noindent where
\begin{equation}
\hat{\varphi}_{kli} = \left. \sum_\sigma \Tr_l \left(
\left(\bm{P}_k\right)_{\bar{\sigma}_{n_R-i,n_R}} \right) \Perm
\left( \left( \bm{Q}_k \right)_{\sigma_{i,n_R}}^{\{k-1\}}\right)
\right., \nonumber
\label{eq:app:macro_capacity:ICj2_denominator:expectation:acute_varphi1}
\end{equation}
and from \dref{eq:macro_capacity:axiom1}, $\hat{\varphi}_{kl0}$
simplifies to give
\begin{equation}
\hat{\varphi}_{kl0} =  \Tr_l \left(\bm{P}_k \right). \nonumber
\label{eq:app:macro_capacity:ICj2_denominator:expectation:acute_varphi_i_zero}
\end{equation}
Equation \dref{eq:app:macro_capacity:ICj2_denominator:expectation:xi3} follows from
\dref{eq:app:macro_capacity:ICj2_denominator:expectation:xi30} due to
the fact that
\begin{equation}
\Tr_l \left( \left(\bm{P}_k\right)_{\bar{\sigma}_{n_R-i,n_R}}
\right) =0 \quad \mbox{for}  \quad l > n_R-i. \nonumber
\end{equation}

Therefore, \dref{eq:app:macro_capacity:ICj2_denominator:expectation} can
be written as
\begin{equation}
\left|\bm{\Sigma}_k\right| E \left\{ \left|\sigma^2 \bm{I} +
\tilde{\bm{H}}_k^H \bm{\Sigma}_k^{-1} \tilde{\bm{H}}_k\right| \right
\} = \sum_{i=0}^{k-1}\sum_{l=0}^{n_R} \left. t \right.^l
\hat{\varphi}_{kli} \left(\sigma^2 \right)^{k-l-i-1}, \nonumber
\end{equation}
\noindent which in turn can be given as
\begin{equation}
\left|\bm{\Sigma}_k\right| E \left\{ \left|\sigma^2 \bm{I} +
\tilde{\bm{H}}_k^H \bm{\Sigma}_k^{-1} \tilde{\bm{H}}_k\right| \right
\} = \sum_{l=0}^{n_R} \left. t \right.^{l} \varphi_{kl}, \nonumber
\label{eq:app:macro_capacity:ICj2_denominator:expectation11}
\end{equation}
\noindent where
\begin{align}
\varphi_{kl} &=  \sum_{i=0}^{k-1} \hat{\varphi}_{kli} \left(\sigma^2
\right)^{k-l-i-1}.
\label{eq:app:macro_capacity:ICj2_denominator:expectation:varphi1}
\end{align}
\section{Extended Laplace Type Approximation}\label{app:macro_capacity:extended_laplace}
Note the well-known fact that, $\sigma^2 \bm{I}= E\left
\{\bm{A}^H\bm{A} \right \}$, for an iid complex Gaussian matrix
ensemble, $\bm{A}$, of $\mathcal{CN}\left(0,\frac{\sigma^2}{\kappa}
\right)$ random variables, where $\bm{A}$ is a $\kappa \times k-1$
matrix as in \cite{Gao98}. This result can be rewritten in the limit
to give $\sigma^2 \bm{I}= \substack { \lim \\ \kappa \rightarrow
\infty } \left \{ \bm{A}^H\bm{A} \right \}$. Using this in
\dref{eq:macro_capacity:general:ICj1} gives
\begin{align}
%
\tilde{I}_{k} \left( t \right) &=  \frac{1}{\left| \bm{\Sigma}_k
\right|} \lim_{\kappa \rightarrow \infty} E \left\{ \frac{\left|
\bm{A}^H\bm{A} + \tilde{\bm{H}}_k^H \tilde{\bm{H}}_k \right|}{\left|
\bm{A}^H\bm{A} + \tilde{\bm{H}}_k^H \bm{\Sigma}_k^{-1}
\tilde{\bm{H}}_k \right|} \right \}
\label{eq:macro_capacity:extended_laplace2},\\
&= \frac{1}{\left| \bm{\Sigma}_k \right|} \lim_{\kappa \rightarrow
\infty} E \left\{ \frac{\left| \left( \bm{A}^H, \tilde{\bm{H}}_k^H
\right) \left( \substack{\bm{A}
\\ \tilde{\bm{H}}_k}\right) \right|}{\left| \left( \bm{A}^H,
\tilde{\bm{H}}_k^H \bm{\Sigma}_k^{-\frac{1}{2}} \right) \left(
\substack{\bm{A} \\ \bm{\Sigma}_k^{-\frac{1}{2}}\tilde{\bm{H}}_k}
\right) \right|} \right \} \label{eq:macro_capacity:extended_laplace3},\\
&= \frac{1}{\left| \bm{\Sigma}_k \right|} \lim_{\kappa \rightarrow
\infty} E \left\{ \frac{\left| \bm{B}_k^H \bm{B}_k \right|}{\left|
\bm{B}_k^H \bar{\bm{\Sigma}}_k \bm{B}_k \right|} \right \}
\label{eq:macro_capacity:extended_laplace4},
\end{align}
\noindent where $\bar{\bm{\Sigma}}_k =\mbox{diag} \left(\bm{I},
\bm{\bm{\Sigma}_k}^{-\frac{1}{2}} \right)$ and $\bm{B}_k=\left(
\substack{\bm{A}
\\ \tilde{\bm{H}}_k}\right)$. Using the well-known fact
\begin{align}
\left| \bm{B}_k^H \bm{B}_k \right|=\prod_{i=1}^{k-1} \bm{b}_{ki}^H
\left(\bm{I} - \tilde{\bm{B}}_{ki} \left(\tilde{\bm{B}}_{ki}^H
\tilde{\bm{B}}_{ki} \right)^{-1} \tilde{\bm{B}}_{ki}^H
\right)\bm{b}_{ki}, \label{eq:macro_capacity:extended_laplace45}
\end{align}
from standard linear algebra, where $\bm{b}_{ki}$ is the $i^{th}$
column of $\bm{B}_k$, we can approximate
\dref{eq:macro_capacity:extended_laplace4} by
\begin{align}
\tilde{I}_{k} \left( t \right) &\simeq \frac{1}{\left| \bm{\Sigma}_k
\right|} \prod_{i=1}^{k-1} E \left\{ \frac{\bm{b}_{ki}^H
\left(\bm{I} - \tilde{\bm{B}}_{ki} \left(\tilde{\bm{B}}_{ki}^H
\tilde{\bm{B}}_{ki} \right)^{-1} \tilde{\bm{B}}_{ki}^H
\right)\bm{b}_{ki}}{\bm{b}_{ki}^H \left(\bar{\bm{\Sigma}}_k -
\bar{\bm{\Sigma}}_k \tilde{\bm{B}}_{ki} \left(\tilde{\bm{B}}_{ki}^H
\bar{\bm{\Sigma}}_k \tilde{\bm{B}}_{ki} \right)^{-1}
\tilde{\bm{B}}_{ki}^H \bar{\bm{\Sigma}}_k \right)\bm{b}_{ki}}
\right\} \label{eq:macro_capacity:extended_laplace5},
\end{align}
\noindent where $\bm{b}_{ki}$ and $\bm{B}_k$ correspond to a large
but finite value of $\kappa$. Approximation
\dref{eq:macro_capacity:extended_laplace5} assumes that the terms in
the product in \dref{eq:macro_capacity:extended_laplace45} are
independent. This is only true when $\bm{b}_{ki}$ contains iid
elements. However, in the macrodiversity case, all the elements of
$\bm{b}_{ki}$ are not iid. Nevertheless, part of $\bm{b}_{ki}$ (the
contribution from $\bm{A}$) is iid.  This motivates the
approximation in \dref{eq:macro_capacity:extended_laplace5}. Next,
we apply the standard Laplace type approximation \cite{Lib94} in
\dref{eq:macro_capacity:extended_laplace5} to give
\begin{align}
\tilde{I}_{k} \left( t \right)&\simeq \frac{1}{\left| \bm{\Sigma}_k
\right|} \prod_{i=1}^{k-1} \frac{ E \left\{ \bm{b}_{ki}^H
\left(\bm{I} - \tilde{\bm{B}}_{ki} \left(\tilde{\bm{B}}_{ki}^H
\tilde{\bm{B}}_{ki} \right)^{-1} \tilde{\bm{B}}_{ki}^H
\right)\bm{b}_{ki}\right \}}{ E \left\{ \bm{b}_{ki}^H
\left(\bar{\bm{\Sigma}}_k - \bar{\bm{\Sigma}}_k \tilde{\bm{B}}_{ki}
\left(\tilde{\bm{B}}_{ki}^H \bar{\bm{\Sigma}}_k \tilde{\bm{B}}_{ki}
\right)^{-1} \tilde{\bm{B}}_{ki}^H \bar{\bm{\Sigma}}_k
\right)\bm{b}_{ki} \right\}}
\label{eq:macro_capacity:extended_laplace6},\\
&\simeq \frac{1}{\left| \bm{\Sigma}_k \right|} \frac{ E \left\{
\prod_{i=1}^{k-1} \bm{b}_{ki}^H \left(\bm{I} - \tilde{\bm{B}}_{ki}
\left(\tilde{\bm{B}}_{ki}^H \tilde{\bm{B}}_{ki} \right)^{-1}
\tilde{\bm{B}}_{ki}^H \right)\bm{b}_{ki}\right \}}{ E \left\{
\prod_{i=1}^{k-1} \bm{b}_{ki}^H \left(\bar{\bm{\Sigma}}_k -
\bar{\bm{\Sigma}}_k \tilde{\bm{B}}_{ki} \left(\tilde{\bm{B}}_{ki}^H
\bar{\bm{\Sigma}}_k \tilde{\bm{B}}_{ki} \right)^{-1}
\tilde{\bm{B}}_{ki}^H \bar{\bm{\Sigma}}_k \right)\bm{b}_{ki}
\right\}}\label{eq:macro_capacity:extended_laplace7} ,\\
&= \frac{1}{\left| \bm{\Sigma}_k \right|} \frac{ E \left\{ \left|
\bm{B}_k^H \bm{B}_k \right| \right\} }{E \left\{ \left| \bm{B}_k^H
\bar{\bm{\Sigma}}_k \bm{B}_k \right|
\right\}}\label{eq:macro_capacity:extended_laplace8}.
\end{align}
Hence, a combination of approximate independence, the Laplace
approximation for quadratic forms and the limiting version in
\dref{eq:macro_capacity:extended_laplace2} gives rise to the
approximation used in Sec. \ref{sec:system_analysis_general_user}.
The accuracy of this approach is numerically established in the
simulation results in Sec.
\ref{sec:macro_capacity:numerical_analysis}.

\begin{IEEEbiography}{Dushyantha Basnayaka}
(S'11) was born in 1982 in Colombo, Sri Lanka. He received the
B.Sc.Eng degree with 1\textsuperscript{st} class honors from the
University of Peradeniya, Sri Lanka, in Jan 2006. He is currently
working towards for his PhD degree in Electrical and Computer
Engineering at the University of Canterbury, Christchurch, New
Zealand.\\
He was an instructor in the Department of Electrical and Electronics
Engineering at the University of Peradeniya from Jan 2006 to May
2006. He was a system engineer at MillenniumIT (a member company of
London Stock Exchange group) from May 2006 to Jun 2009. Since Jun.
2009 he is with the communication research group at the University
of Canterbury, New Zealand.\\
D. A. Basnayaka is a recipient of University of Canterbury
International Doctoral Scholarship for his doctoral studies at UC.
His current research interest includes all the areas of digital
communication, specially macrodiversity wireless systems. He holds
one pending US patent as a result of his doctoral studies at UC.
\end{IEEEbiography}
\begin{IEEEbiography}{Peter Smith}
(M'93-SM'01) received the B.Sc degree in Mathematics and the Ph.D
degree in Statistics from the University of London, London, U.K., in
1983 and 1988, respectively. From 1983 to 1986 he was with the
Telecommunications Laboratories at GEC Hirst Research Centre. From
1988 to 2001 he was a lecturer in statistics at Victoria University,
Wellington, New Zealand. Since 2001 he has been a Senior Lecturer
and Associate Professor in Electrical and Computer Engineering at
the University of Canterbury in New Zealand. His research interests
include the statistical aspects of design, modeling and analysis for
communication systems, especially antenna arrays, MIMO, cognitive
radio and relays.
\end{IEEEbiography}
\begin{IEEEbiography}{Philippa Martin}
(S'95-M'01-SM'06) received the B.E. (Hons. 1) and Ph.D. degrees in
electrical and electronic engineering from the University of
Canterbury, Christchurch, New Zealand, in 1997 and 2001,
respectively.  From 2001 to 2004, she was a postdoctoral fellow,
funded in part by the New Zealand Foundation for Research, Science
and Technology (FRST), in the Department of Electrical and Computer
Engineering at the University of Canterbury.  In 2002, she spent 5
months as a visiting researcher in the Department of Electrical
Engineering at the University of Hawaii at Manoa, Honolulu, Hawaii,
USA.  Since 2004 she has been working at the University of
Canterbury as a lecturer and then as a senior lecturer (since 2007).
In 2007, she was awarded the University of Canterbury, College of
Engineering young researcher award.  She served as an Editor for the
IEEE Transactions on Wireless Communications 2005-2008 and regularly
serves on technical program committees for IEEE conferences.  Her
current research interests include multilevel coding, error
correction coding, iterative decoding and equalization, space-time
coding and detection, cognitive radio and cooperative communications
in particular for wireless communications
\end{IEEEbiography}

\end{document}